\documentclass[article]{IEEEtran}
\usepackage[utf8]{inputenc}
\usepackage{amsmath,amsthm,amsfonts,amssymb,bm,accents,dsfont}
\usepackage{graphicx,cite,enumerate}
\usepackage{url}
\usepackage{caption}
\usepackage{subcaption}
\graphicspath{{figures/}}

\usepackage{blindtext}

\newcommand{\executeiffilenewer}[3]{%
\ifnum\pdfstrcmp{\pdffilemoddate{#1}}%
{\pdffilemoddate{#2}}>0%
{\immediate\write18{#3}}\fi%
}

\usepackage[english]{babel}

\usepackage{algpseudocode,algorithm}
\algrenewcommand\algorithmicdo{}
\algrenewtext{EndFor}{\algorithmicend}
\algrenewtext{EndProcedure}{\algorithmicend}

\usepackage{color}
\usepackage{mysymbol}

\newtheorem{theorem}{Theorem}
\newtheorem{proposition}{Proposition}

\newtheorem{lemma}{Lemma}
\newtheorem{hyp}{Hypothesis}
\theoremstyle{definition}

\newtheorem{remark}{Remark}
\newtheoremstyle{assume}
  {3pt}
  {3pt}
  {}
  {}
  {\bf}
  {}
  { }
  {\thmname{#1}.\thmnumber{#2}\thmnote{ \textnormal{(\textit{#3})}}}
\theoremstyle{assume}

\def\bblam{\boldsymbol{\lambda}}

\newcommand{\ubar}[1]{\ensuremath{\underaccent{\bar}{#1}}}

\def\nnil{\nil}
\newcounter{prob}
\newenvironment{prob}[1][\nil]{%
	\def\tmp{#1}
	\equation
	\ifx\tmp\nnil
		\refstepcounter{prob}
		\tag{P\Roman{prob}}
	\else
		\tag{\tmp}
	\fi
	\aligned%
}{%
	\endaligned\endequation%
}

\usepackage{enumerate}
\usepackage{caption}
\usepackage{rotating}
\title{Federated Classification using Parsimonious Functions in Reproducing Kernel Hilbert Spaces}
\author{Maria Peifer and Alejandro~Ribeiro%
\thanks{Department of Electrical and Systems Engineering, University of Pennsylvania.
e-mail: \mbox{\texttt{mariaop@seas.upenn.edu}}~(contact author), \mbox{\texttt{aribeiro@seas.upenn.edu}}.
}
}

\renewcommand {\st} {\text{\,s.t.}}

\begin{document}
\maketitle
\begin{abstract}
Federated learning forms a global model using data collected from a federation agent. This type of learning has two main challenges: the agents generally don't collect data over the same distribution, and the agents have limited capabilities of storing and transmitting data. Therefore, it is impractical for each agent to send the entire data over the network. Instead, each agent must form a local model and decide what information is fundamental to the learning problem, which will be sent to a central unit. The central unit can then form the global model using only the information received from the agents. We propose a method that tackles these challenges. First each agent forms a local model using a low complexity reproducing kernel Hilbert space representation. From the model the agents identify the fundamental samples which are sent to the central unit. The fundamental samples are obtained by solving the dual problem. The central unit then forms the global model. We show that the solution of the federated learner converges to that of the centralized learner asymptotically as the sample size increases. The performance of the proposed algorithm is evaluated using experiments with both simulated data and real data sets from an activity recognition task, for which the data is collected from a wearable device. The experimentation results show that  the accuracy of our method converges to that of a centralized learner with increasing sample size. 
\end{abstract}

\section{Introduction}\label{S:intro}	
In federated learning a global model is trained by a central server from a federation of agents \cite{konevcny2016fed,konevcny2016federated,li2020federated, mcmahan2016communication}. Different from traditional learning, in which the learner has access to the entire data set, in federated learning each agent collects its own data. A naive method for solving this problem involves having all agents send their data over the network to the central server. The serve then computes a global model using traditional machine learning methods. This is a feasible method when the number of agents and the size of the data collected is small. However, it becomes increasingly prohibitive as the agents collect more data and it does not take advantage of the computational capabilities of the agents.

The need for federated learning algorithms naturally arose from distributed networks generating a vast amount of data such as mobile phones, wearable devices, and autonomous vehicles \cite{smith2017federated, anguita2013public}. These devices are not only capable of collecting large amounts of data but also have great computational capabilities which makes them an essential part of the learning process. Applications of federated learning include learning sentiment, semantic location, or activities of mobile phone users; predicting health events like low blood sugar or heart attack risk from wearable devices; or detecting burglaries within smart homes \cite{smith2017federated,pantelopoulos2009survey,rashidi2009keeping,anguita2013public,yang2019federated}.

Federated learning presents two major challenges: statistical challenge and system challenge \cite{konevcny2015federated, konevcny2016federated,smith2017federated}. The statistical challenge arises from the fact that an agent is tied generally tied to a user, \emph{e.g.} mobile device user, and while there exists similarities between users, each individual is different and therefore the distribution of the data collected by the agents is different. Indeed, the data collected by the agents is non i.i.d., although in some cases the organization of the agents can provide information about the relationship between the distributions \cite{zhao2018federated}. The system challenge is dictated by the capabilities of the system to store data and to transmit data over the network efficiently \cite{bonawitz2019towards, konevcny2016federated, smith2017federated}. Devices collect vast amounts of data and therefore transmitting the data over the network is prohibitive. Consequently, a traditional learning approach with a central unit learning the global model is often impossible, and it is imperative for agents to transmit information about the problem without sending their entire data.

The task of learning from a federation of agents has been approached several ways. The  algorithm \emph{FederatedAveraging (FedAvg)} was proposed  in \cite{mcmahan2017federated, hard2018federated} which computed a weighted average of the models from each agent based on the number of training samples at each update step. An equivalent algorithm for \emph{FedAvg} was proposed in \cite{mcmahan2016communication} which used a weighted average of the gradient to update the global parameters. The \emph{FedAvg} was shown to converge in non i.i.d. settings when used with a diminishing learning rate  for strongly convex and smooth problems \cite{li2019convergence}. A modified \emph{FedAvg} algorithm was later proposed which only performs global updates at a set interval \cite{wang2019adaptive,yu2019parallel}, which was shown to converge near the global optimum. An adaptation to federated learning of the stochastic variance reduced gradient descent was used in \cite{konevcny2016fed} in order to compute the global gradient. In an effort to reduce the communication load even further and increase privacy a stochastic gradient descent variation was used in \cite{shokri2015privacy}, which only transmits a subset of the gradient to the server. The methods presented so far attempt to address the communication challenge and are mostly tested on independent identically distributed (i.i.d.) data sets. Zhao \emph{et.al} \cite{zhao2018federated} attempts to tackle the statistical challenge of non i.i.d data by sharing a small subset of data over the network. 

In this work we address the challenges of federated learning with an algorithm using low complexity reproducing kernel Hilbert space (RKHS) representations. Each agent locally learn a low complexity model and sends only the critical samples to the learning problem over the network to the central server. RKHS methods are non-parametric techniques used in signal processing, statistics and machine learning \cite{rosipal2001kernel, bishop2006pattern, yuan2010reproducing, berlinet2011reproducing, arenas2013kernel, koppel2019parsimonious}. Their success is due to the richness of these spaces, which allows them to model a large class of functions. Moreover, the functions found in RKHS methods can be represented as a possibly infinite linear combination of functions called reproducing kernels \cite{bishop2006pattern,berlinet2011reproducing,peifer2020sparse}. The integral representation of RKHS methods  presented in \cite{peifer2020sparse} is used for this federated classification because it both learns a low complexity representation and detects the critical samples to the learning problem. Indeed, the integral representation allows for a modified $L_0$-norm for functions which measures the support of the function. This $L_0$-norm allows the method to find a low complexity representation of a function in RKHS. Additionally, solving the problem in the dual domain, provides insight about the critical samples to the problem through the dual variable. Our federated classification model requires each agent to train a local model and find the optimal dual variable. The samples for which the dual variable is greater than zero are considered critical to the classification problem and are sent to the central server along with the optimal dual variable. The central server then computes a global model using the samples provided by the agents.

The paper is structured as follows. In section \ref{sec:learning_RKHS} we present a method for learning low complexity representations in RKHSs. Section \ref{sec:prob_form} introduces the problem of federated learning and defines the relationship between the agents. In section \ref{sec:Algorithm} we describe the algorithm for the federated classification problem. In section \ref{sec:optimality} we show that the solution of the federated learner converges to that of the centralized learner as the sample size grows. Lastly, in section \ref{sec:Num} we evaluate our method on both simulated data and a classification task using data obtained from wearable devices.

\section{Learning Low complexity RKHS Representations}
\label{sec:learning_RKHS}
We study classification problems in which we are given a training set made up of $N$ feature-class pairs of the form $(\bbx_n,y_n) \in \ccalT$. Features $\bbx_n \in \calX\subset \reals^p$ are real valued $p$-dimensional vectors and classes $y_n\in\{-1, +1\}$ are binary. We want to learn a function approximation $f(x)$ such that $f(x_n)$ coincides with $y_n$ to the extent possible. We formulate this mathematically by introducing the loss function $\ell(f(x), y) = 1-\epsilon - y f(x)$, a class function $\ccalC$ and a function complexity measure $\rho(f)$ to define the optimization problem

\begin{prob}\label{eqn_og}
 P=  &\min_{f\in\ccalC} 
        && \rho(f) \\
   &\st
       &&  \frac{1}{N}\ell\Big(f(\bbx_n), y_n\Big) \leq 0, \quad
                (\bbx_n,y_n) \in \ccalT .
\end{prob}
The constraints in \eqref{eqn_og} force the function $f$ to satisfy $f(\bbx_n)\geq 1-\eps$ when $y_n=+1$ and $f(\bbx_n)\leq -(1-\eps)$ when $y_n=-1$. The class function $\ccalC$ and the complexity measure $\rho(f)$ constrain the variability of $f$ and dictate how it generalizes to unobserved samples $\bbx$. In this paper we adopt the use of low complexity kernel representations as introduced in \cite{peifer2020sparse}.

To formalize low complexity RKHS representations, let $k(\bbx, \bbs; w)$ be a \emph{family} of kernel functions in which $\bbx\in\reals^p$ is a variable, $\bbs\in\reals^p$ is a kernel center and $w\in\reals$ is a kernel parameter. Further consider a compact set of possible kernel parameters $\ccalW\subset\reals$, and a compact set of possible kernel centers $\ccalS\subseteq\reals^p$. The function class $\ccalC$ is defined as 
\begin{equation}\label{eqn_integral_rep}
   \ccalC  
      = \bigg\{f \,:\, 
            f(\bbx) = \int_{\ccalS \times \ccalW} 
                         \alpha(\bbs,w) 
                             k(\bbx,\bbs \,;\, w) \, d\bbs dw 
        \bigg\},
\end{equation}
where $\alpha: \ccalS \times \ccalW \rightarrow \reals$ is a coefficient function in~$L_2(\ccalS \times \ccalW)$. Observe that in \eqref{eqn_integral_rep} the kernel family $k(\bbx,\bbs \,;\, w)$  is given and the coefficient function $\alpha(\bbs,w)$ is a variable we want to find. We leverage this observation to measure the complexity of the function $f\in\ccalC$ through the complexity of $\alpha$. We begin by defining the sparsity of the coefficient function as its support, which is given by
\begin{equation}\label{eqn_L0}
   \Vert \alpha \Vert_{L_0} = \int_{\ccalS \times \ccalW}
		\indicator \left[ \alpha(\bbs, w) \neq 0 \right] d\bbs dw.
\end{equation} 
Further consider the $L_2$ norm  $\Vert \alpha \Vert_{L_2} = \int_{\ccalS \times \ccalW}  \alpha^2(\bbs, w)  d\bbs dw$ of the coefficient function so as to measure the complexity of $f$ according to the elastic net measure 
\begin{align}\label{eqn_elastic_net}
   \rho(f) 
      &= \frac{1}{2} \Vert \alpha \Vert_{L_2}  + 
              \gamma \Vert \alpha \Vert_{L_0}  \\ \nonumber
      &= \int_{\ccalS \times \ccalW}
          \frac{1}{2} \alpha^2(\bbs, w) +
		           \gamma \indicator \left[ \alpha(\bbs, w) \neq 0 \right] \,d\bbs dw.
\end{align} 
The variable $\gamma$ represents the regularizing parameter between the two norms.
In principal, we want $\gamma$ as large as possible to favor sparsity. In practice, however, we want to reduce the value to benefit from the regularizing effect (see Remark \ref{rem_gamma}). 
The low complexity RKHS classification problem is defined as \eqref{eqn_og} with  the class function $\ccalC$ given by \eqref{eqn_integral_rep} and the function complexity measure given by \eqref{eqn_elastic_net}. This is a problem that we can rewrite as an optimization over the coefficient function,
\begin{prob}[PC]\label{eqn_the_central_problem}
P = &\min_{\alpha}
	   &\int_{ \ccalS \times \ccalW}
          \frac{1}{2} \alpha^2(\bbs, w) +
		           \gamma \indicator \left[ \alpha(\bbs, w) \neq 0 \right] \,d\bbs dw \\
   &\st &  \frac{1}{N}\ell\Big(f(\bbx_n), y_n\Big) \leq 0, \quad
                (\bbx_n,y_n) \in \ccalT  \\
   &    &  f(\bbx) = \int_{\ccalS \times \ccalW} 
                         \alpha(\bbs,w) 
                             k(\bbx,\bbs \,;\, w) \, d\bbs dw, \\
\end{prob}
The problem \eqref{eqn_the_central_problem} differs from \eqref{eqn_og} in that it replaces the search for the function $f$ by a search for the coefficient $\alpha$. The problems are otherwise equivalent -- with function class $\ccalC$ as per  \eqref{eqn_integral_rep} and complexity measure $\rho$ as per \eqref{eqn_elastic_net} -- in the sense that both attain the same optimal objective $P$ and we can recover the optimal function $f^*$ from the the optimal coefficient $\alpha^*$ by evaluating $f(\bbx)$ according to \eqref{eqn_integral_rep}. 

The constraints in \eqref{eqn_the_central_problem} specify the form of the function $f$ in terms of the coefficient function $\alpha$ and force the constraints $\ell(f(\bbx_n), y_n) \leq 0$ to be satisfied for all entries of the training set. Out of all the coefficient functions that satisfy these constraints we search for the $\alpha^*$ with the lowest elastic net cost. This is expected to be a sparse coefficient function. We are therefore searching for a function $f$ that can be specified by as few kernels as possible while still passing within $\eps$ of all the elements of the training set. When we search for kernels to add to the representation, the search is over kernel centers $\bbs\in\ccalS$ and kernel parameters $\bbw\in\ccalW$. The latter allows, e.g., a search over kernel widths -- see \cite{peifer2020sparse} for details.

The problem \eqref{eqn_the_central_problem} is infinite dimensional and non-convex, however, it can be solved in the dual domain by leveraging the fact that it has zero duality gap \cite{chamon2018strong,peifer2020sparse}. 

\begin{theorem}\label{th:zero_duality}
The problem \eqref{eqn_the_central_problem} has strong duality if the following conditions are met: (i) the kernel $k(\cdot,\bbs;w)$ has no point masses and (ii) Slater's condition is met.
\end{theorem}

A formal proof for this theorem can be found in \cite{peifer2020sparse}. 

\begin{remark}\normalfont
Problem \eqref{eqn_the_central_problem} fits the classification function by using constraints as opposed to a regularized minimization problem. The advantage of using a constraint problem is two-fold: it allows for the problem to be solved in the dual domain and the solution of the dual gives us information about the critical samples to our learning problem. This concept is explained in more detail in sections \ref{sec:Algorithm} and \ref{sec:optimality}.
\end{remark}

\begin{remark}\normalfont
The function class $\ccalC$ is not an RKHS but is closely related. For a fixed kernel parameter $w$ in \eqref{eqn_integral_rep} the expression $f(\bbx) = \int_{\ccalS} \alpha(\bbs,w) k(\bbx,\bbs; w) d\bbs$ is an integral representation of the RKHS generated by the kernel $k(\bbx,\bbs; w)$ \cite{peifer2018locally, peifer2019sparse, peifer2020sparse}. The use of this integral representation as opposed to the more traditional series representation $f(\bbx) = \sum_{j=1}^J a_j k(\bbx, \bbs_j;w)$ may seem an unnecessary complication but it is actually a crucial simplification. The search for a sparse set of kernels is intractable with a series representation. But the problem in \eqref{eqn_the_central_problem} is tractable in the dual domain. As a byproduct of this more tractable formulation, we can also incorporate the kernel parameter $w$ to the search space and still guarantee tractability. This results in a problem that not only optimizes kernel placement but also kernel width and can even accommodate representations in unions of RKHS, i.e., representations having a mix of kernels of different widths \cite{peifer2020sparse}. 
\end{remark}

\section{Federated Learning}
\label{sec:prob_form}
In the previous section, we have presented learning in a centralized setting. In this section, we illustrate learning in a federated setting in which
a centralized method for obtaining low complexity reproducing kernel Hilbert space  data is collected by a group of agents over the space $\calX$. The federation of agents must work together to find a global model over $\calX$. To this end, the federation adopts the strategy of each agent learning a local model using the data it collects. From that model, the agent detects the critical samples to the classification problem. The agent sends only the critical samples to the central server which learns the global model.
Particularly, given a set of $N_i$ feature-class pairs of the form $(\bbx_n, y_n) \in \ccalT_i$, agent $A_i$ solves the following problem
\begin{prob}[Pi]\label{eqn_the_agent_problem}
	P_i = &\min_{\alpha \in L_2}
		&&\int_{ \ccalS \times \ccalW}
          \frac{1}{2} \alpha^2(\bbs, w) +
		          \gamma \indicator \left[ \alpha(\bbs, w) \neq 0 \right] \,d\bbs dw
	\\
	& \st && \frac{1}{N_i} \ell(f(\bbx_n), y_n) \leq 0 \text{,} \quad  (\bbx_n,y_n) \in \ccalT_i\\
	&    &&  f(\bbx) = \int_{\ccalS \times \ccalW} 
                         \alpha(\bbs,w) 
                             k(\bbx,\bbs \,;\, w) \, d\bbs dw.
\end{prob}
In order to find the set of critical feature-class pairs $\tilde{\calT}_i \subset \calT_i$ and a model parameter to send to the central server. The central unit learns the problem using $\tilde{\calT} = \cup_i \tilde{\calT}_i$ such that $\vert \calT\vert \gg \vert \tilde{\calT}\vert $, where $\calT = \cup_i \calT_i$. 
Typically, each agent $A_i$ is not able to sample $\ccalX$ entirely, but rather observes a subspace $\calX_i$, however, the subspaces, observed by the agents, cover the space $\ccalX$, such that $\cup_i \calX_i = \calX$. 

Notice, problems \eqref{eqn_the_central_problem} and \eqref{eqn_the_agent_problem} are minimizing the same objective function, however, \eqref{eqn_the_central_problem} has additional constraints due to a larger sample set. Problem \eqref{eqn_the_agent_problem} is limited to only samples from a specific subspace. Although this might seem like an initial disadvantage, solving a smaller problem can improve computational speed, whereas, solving \eqref{eqn_the_central_problem} requires a lot of information sharing from each agent which can become impractical. Moreover, by solving the dual problem of \eqref{eqn_the_agent_problem}, we can obtain the critical samples of the classification problem. The central server uses the critical samples from the agents to find the global model. Formally, the central server solves the problem
\begin{prob}[PF]\label{eqn_the_feder_problem}
	PF = &\min_{\alpha}&
		&\int_{ \ccalS \times \ccalW}
          \frac{1}{2} \alpha^2(\bbs, w) +
		          \gamma \indicator \left[ \alpha(\bbs, w) \neq 0 \right] \,d\bbs dw
	\\
	&\st && \frac{1}{N_F} \ell(f(\bbx_n), y_n) \leq 0 \text{,} \quad  (\bbx_n, y_n) \in \tilde{\calT}\\
	&    &&  f(\bbx) = \int_{\ccalS \times \ccalW} 
                         \alpha(\bbs,w) 
                             k(\bbx,\bbs \,;\, w) \, d\bbs dw
		\text{,}
\end{prob}
where $N_F$ is the number of critical samples in $\tilde{\calT}$.
Notice, problems \eqref{eqn_the_central_problem} and \eqref{eqn_the_feder_problem} solve the same problem, however \eqref{eqn_the_feder_problem} solves it for a restricted data set. 
The goal is to find a subset $\tilde{\calX}$ such that the solution of \eqref{eqn_the_feder_problem} is close to that of \eqref{eqn_the_central_problem}. The simplest solution is to make $\tilde{\calT} = \calT$, by pooling the data collected from all the agents and have the central server compute the global model. However, this solution involves a large amount of data to be sent which could surpass the capabilities of the network.
In the next section, we present the algorithm for obtaining the solution to our federated learning problem.

\section{Algorithm}
\label{sec:Algorithm}
The federated classification problem requires the agents to solve their local problem \eqref{eqn_the_agent_problem} in order to find a local model and detect the critical samples. The critical samples to the classification problem are sent to the server. The server then forms the global model by solving \eqref{eqn_the_feder_problem}. The federated classification algorithm is summarized in Algorithm \ref{alg}. 
The next section describes the algorithm for solving the agent problem.
\begin{algorithm}[t]
\caption{Federated classification algorithm}\label{alg}
\begin{algorithmic}
	\For i 
	\State agent $i$ samples the subspace $\calX_i$
	\State agent $i$ solves \eqref{eqn_the_agent_problem} and calculate optimal $\bblam_i^*$
	\State agent $i$ sends critical samples for which $\bblam_{i,n}^* > 0$
	\EndFor

	\State central unit solves \eqref{eqn_the_feder_problem}
\end{algorithmic}

\end{algorithm}

\subsection{The Agent Problem}
The agent solves problem \eqref{eqn_the_agent_problem} in the dual domain. In order to derive the dual problem, agent $i$ defines the Lagrange multiplier $\bblam_i \in \reals_+^{N_i}$, associated with the inequality constraints. Formally, the Lagrangian  is characterized as
\begin{equation} \label{eqn_Lagrange_agent}
\begin{aligned}
    \calL_i(\alpha, \bblam_i) &= \int_{ \ccalS \times \ccalW}
          \frac{1}{2} \alpha^2(\bbs, w) +
		           \gamma \indicator \left[ \alpha(\bbs, w) \neq 0 \right] \,d\bbs dw \\
   & + \frac{1}{N_i} \sum_n^{N_i} \bblam_{i,n} \ell(f(\bbx_{i,n}), y_{i,n}).
\end{aligned}
\end{equation}
Each element of the Lagrangian multiplier is associated with the loss over a single sample point. The Lagrangian is less than or equal to the primal function for any feasible $\alpha$. Therefore, by minimizing the Lagrangian over $\alpha$ each agent obtains a lower bound for the primal problem. This is called the dual function
\begin{equation}\label{eqn_dual_agent}
    g_i(\bblam_i) = \min_{\alpha \in L_2} \calL_i(\alpha, \bblam_i).
\end{equation}
The dual function is the minimum over a set of affine functions of $\bblam$ and is therefore concave \cite{boyd2004convex}. 
Additionally, it is a lower bound to the primal problem for any feasible function $\alpha$ which meets the constraints. Indeed, the dual function is the sum of the primal function and the constraints weighted by the Lagrangian multiplier. In order for a function $\alpha$ to be feasible, the constraints must be non-positive and therefore the dual function can be at most equal to the primal function.
Maximizing the dual function results in the best lower bound. Moreover, when strong duality holds the maximum value of the dual function is equal to the solution of the primal problem. This leads to the formulation of the dual problem
\begin{prob}[Di]\label{eqn_the_dual_agent_problem}
    \begin{aligned}
    &\maximize_{\bblam_i \geq 0} & g_i(\bblam_i).
    \end{aligned}
\end{prob}
Solving the dual problem provides a solution for the primal problem \cite{peifer2020sparse}.
Because the dual function is concave, the dual problem can be solved using gradient descent \cite{boyd2004convex}. The gradients can be obtained by evaluating the constraints at $\alpha_d$, which minimizes the Lagrangian
\begin{equation}\label{eqn_alpha_i}
\alpha_i^\star(\bbs,w, \bblam_i) = \underset{\alpha\in L_2}{\text{argmin}} ~\calL_i(\alpha,\bblam_i).
\end{equation}
In order to find $\alpha_i^\star$ we must minimize the function $\calL_\alpha$, the term of the Lagrangian which depends on $\alpha$ 
\begin{equation}\label{eqn_lagr_alpha}
\begin{aligned}
    \calL_\alpha(\alpha, \bblam_i) &=\int \Bigg[\frac{1}{2} \alpha^2(\bbs,w) + \gamma \indicator[\alpha(\bbs,w) \neq 0] \\ 
   & -\frac{1}{N_i}\sum_n^{N_i} \bblam_{i,n} y_n \alpha(\bbs,w) k(\bbx_n, \bbs;w) \Bigg] d\bbs dw.
\end{aligned}
\end{equation}
The function in \eqref{eqn_lagr_alpha} can be minimized with respect to $\alpha$ for each variable $\bbs$ and $w$ separately \cite{peifer2020sparse}. Therefore, the minimization of $\calL_\alpha$ reduces to the minimization of a quadratic function with a discontinuity at $\alpha=0$ and hence, we obtain a closed-form thresholding solution of $\alpha_d(\bbs,w)$
\begin{equation}\label{eqn_alpha_i_st}
\alpha_i^\star(\bbs, w;\bblam) = \begin{cases}  \bar{\alpha}_i(\bbs, w;\bblam)&  (\bar{\alpha}_i(\bbs, w;\bblam))^2 >2\gamma \\
0 & otherwise,
\end{cases} 
\end{equation}
for which
\begin{equation}
\bar{\alpha}_i(\bbs, w;\bblam) = \frac{1}{N_i} \sum_n \bblam_{i,n} y_n k(\bbs, \bbx_n,w).
\end{equation}
The gradient has the following expression
\begin{equation}\label{eqn_gradient_i}
    d_{\bblam_{i,n}} = \nabla_{\bblam_{i,n}} g_i(\bblam_i) = \frac{1}{N_i} \ell(f_d(\bbx_n),y_n),
\end{equation}
where $f_{i}^\star$ is given by:
\begin{equation}
    f_{i}^\star = \int_{\calX \times \calW} \alpha_i^\star(\bbs,w) k(\bbx,\bbs;w) ~d\bbs dw.
\end{equation}
Each agent starts by initializing the dual variable $\bblam_i \in \reals^{N_i}_+$ to a positive random value. The gradient of the dual function provides the direction of descent, however, it does not provide any information on how close we are to the maximum, nor does it provide any information about how long to move along that direction. Therefore, a small step size $\eta$ to move along the gradient such that the direction of descent is evaluated often. The variable is updated in the direction of the gradient as follows
\begin{equation}\label{eqn_update_i}
    \bblam_i(t+1) = \left[ \bblam_i(t) + \eta ~d(\bblam_i) \right]_+ ,
\end{equation}
where $\left[m\right]_+ = max(0,m)$. The dual problem is constrained to only have non-negative values for $\bblam_i$ and therefore the updates are restricted. 

Once the gradient descent algorithm has converged, the critical samples are identified by examining the optimal dual variable $\bblam_i^*$.
Notice that the Lagrangian is the primal function to which the constraints are added weighted by the Lagrange multiplier $\bblam_i$. For feasible $\alpha$ the constraints are always non-positive. Moreover, at the optimal dual variable, the constraints multiplied by the optimal dual variable have to be zero  for strong duality to hold which means that either the constraints or the dual variable are equal to zero. This is known as complementary slackness \cite{boyd2004convex}. Hence the dual variable is an indicator that certain constraints are difficult to satisfy:
\begin{equation}
\begin{cases}
    1-\epsilon -y_n\hat{y}_n = 0,&  \bblam >0 \\
    1-\epsilon -y_n\hat{y}_n < 0,&  \bblam =0.
    \end{cases}
\end{equation}

The solution to the primal problem \eqref{eqn_the_agent_problem}, $\alpha_i^*(\bbs,w)$, can be found according to the following proposition
\begin{proposition}\label{prop:optimal_primal}
Let $\bblam_i^*$ be the solution of \eqref{eqn_the_dual_agent_problem}, then the solution to problem \eqref{eqn_the_agent_problem} is given by $\alpha_i^\star(\cdot,\cdot, \bblam_i^*)$ from \eqref{eqn_alpha_i_st}.
\end{proposition}
A formal proof can be found in \cite{peifer2020sparse}. Proposition \ref{prop:optimal_primal} suggests that the solution to problem \eqref{eqn_the_agent_problem}, $\alpha_i^*(\bbs,w)$ is a weighted sum of kernels centered at the sample points. Samples, for which the dual variable $\bblam_{i,n}^* = 0$, do not contribute to the function $\alpha_i^*$ and therefore are not considered critical to the problem.
\begin{algorithm}[t]
\caption{Agent algorithm}\label{alg_agent}
\begin{algorithmic}
\State Collects data over subspace $\calX_i$
\State Initialize $\bblam_i(0) > 0$ randomly
	\For $t = 0, \dots, T$
	\State Compute $\alpha_i(\bbs,w, \bblam_i)$
	\begin{equation*}
		\alpha_i^\star(\bbs, w, \bblam_i) = \begin{cases}
			\bar{\alpha}_i(\bbs, w, \bblam_i)
			\text{,} &\left\vert \bar{\alpha}_i(\bbs,w, \bblam_i) \right\vert > \sqrt{2\gamma}
			\\
			0 \text{,} &\text{otherwise}
		\end{cases}
	\end{equation*}
	for ~$\bar{\alpha}_i(\bbs, w;\bblam_i) = \frac{1}{N_i} \sum_n \bblam_{i,n} y_n k(\bbs, \bbx_n,w)$
	\State evaluate the gradient 
	\begin{equation*}
	    d_{\bblam_n} (t) = \frac{1}{N_i} \ell(f_d(x_n),y_n)
	\end{equation*}
	for which
	\begin{equation*}
	    f_d(x_n) = \int_{\calX\times \calW} \alpha_d(\bbs,w) k(\bbx_n,\bbs;w) d\bbs dw
	\end{equation*}
	\State Update local parameter  
	\begin{equation*}
	    \bblam_{i,n}(t+1) = \left[\bblam_{i,n}(t) + \eta~ d_{\bblam_n} (t) \right]_+
	\end{equation*}
	\EndFor
	\State Let the local optimal dual variable be $\bblam_i^\star = \bblam_i(t+1)$
	\State Determine critical sample pairs: $\tilde{\calT}_i = \{ (\bbx_n, y_n) \mid \bblam_n^*>0\}$
	\State Send $\Tilde{\calT}_i$ and  $\tilde{\bblam}_i^* = \{ \bblam_n^* \mid \bblam_n^*>0\}$ to the server
\end{algorithmic}
\end{algorithm}

\subsection{The Server Problem} 
The server receives the critical sample pairs $\tilde{\calT}_i$ from each agent along with the optimal dual variables $\tilde{\bblam}_i^*$ and  forms the training its set $\tilde{\calT} = \cup_i \tilde{\calT}_i$ which is used to solve problem \eqref{eqn_the_feder_problem} in the dual domain. The Lagrange multiplier $\bblam_F \in \reals_+^{N_F}$ is defined in order to formulate the Lagrangian of \eqref{eqn_the_feder_problem} 
\begin{equation}\label{eqn_lagrang_feder}
\begin{aligned}
    \calL_F(\alpha, \bblam_F) &= \int_{ \ccalS \times \ccalW}
          \frac{1}{2} \alpha^2(\bbs, w) +
		           \indicator \left[ \alpha(\bbs, w) \neq 0 \right] \,d\bbs dw \\
   & + \frac{1}{N_F} \sum_n^{N_F} \bblam_{F,n} \ell(f(\bbx_{n}), y_{n}).
\end{aligned}
\end{equation}
In a similar manner to \eqref{eqn_dual_agent} and \eqref{eqn_the_dual_agent_problem} the dual function $ g_F(\bblam_F)$ and the corresponding dual problem are established.
The server solves its dual problem using gradient descent. The gradients are computed by evaluating the constraints of \eqref{eqn_the_feder_problem} at $\alpha_F^\star = \argmin_{\alpha} \calL_F(\alpha,\bblam_F)$
\begin{equation}\label{eqn_gradient}
    d_{\bblam_F ,n} = \nabla_{\bblam_n} g_F(\bblam_F) = \frac{1}{N_F} \ell(f^\star(\bbx_n),y_n),
\end{equation}
where $f^\star$ is given by:
\begin{equation}
    f^\star = \int_{\calX \times \calW} \alpha_F^\star(\bbs,w) k(\bbx,\bbs;w) ~d\bbs dw.
\end{equation}
At the beginning of the algorithm, the server initializes the dual variable $\bblam_F(0) = \left[\tilde{\bblam}^*_1, \dots \tilde{\bblam}_K\right]$. Then for each iteration $t$,  $\alpha_F^\star(\bbs,w,\bblam_F(t))$ is computed and used to find the gradient according to \eqref{eqn_gradient}. A small step size $\eta_F$ is chosen to move along the gradient and update the variables
\begin{equation}\label{eqn_update}
    \bblam_F(t+1) = \left[ \bblam_F(t) + \eta_F ~ d_{\bblam_F}(t) \right]_+ ,
\end{equation}
where $\left[m\right]_+ = max(0,m)$.  

\begin{algorithm}[t]
\caption{Server algorithm}\label{alg_server}
\begin{algorithmic}
    \State Receive $\tilde{\calT}_i$ from all agents $i = 1,\dots,K$
    \State Initialize $\bblam_F(0) = [\tilde{\bblam}_1, \dots, \tilde{\bblam}_K]$
	\For $t = 0, \dots, T_F$
	\State Compute $\alpha_F^\star(\bbs,w)$
	\begin{equation*}
		\alpha_F^\star(\bbs, w,\bblam_F) = \begin{cases}
			\bar{\alpha}_F(\bbs, w, \bblam_F)
			\text{,} &\left\vert \bar{\alpha}_F(\bbs,w, \bblam_F) \right\vert > \sqrt{2\gamma}
			\\
			0 \text{,} &\text{otherwise}
		\end{cases}
	\end{equation*}
	for ~$\bar{\alpha}_F(\bbs, w;\bblam_F) = \frac{1}{N_F} \sum_n \bblam_{F,n} y_n k(\bbs, \bbx_n,w)$
	\State evaluate the gradient 
	\begin{equation*}
	    d_{\bblam_{F,n}} (t) = \frac{1}{N_F} \ell(f^\star(x_n),y_n)
	\end{equation*}
	for which
	\begin{equation*}
	    f^\star(x_n) = \int_{\calX\times \calW} \alpha_F^\star(\bbs,w) k(\bbx_n,\bbs;w) d\bbs dw
	\end{equation*}
	\State Update local parameter  
	\begin{equation*}
	    \bblam_{F,n}(t+1) = \left[\bblam_{F,n}(t) + \eta~ d_{\bblam_F,n} (t)\right]_+
	\end{equation*}
	\EndFor
    \State Compute global $\alpha^*(\bbs,w) = \alpha^\star(\bbs,w,\bblam_F(T_F+1))$
	\State Send global model to the agents
\end{algorithmic}

\end{algorithm}

\section{Convergence of federated problem}
	\label{sec:optimality}
In the previous section, we have presented a federated learning problem and proposed a method for each agent to solve a local problem and transmit a set of critical samples to a central server which in turn produces a global model. In this section, we argue that solving \eqref{eqn_the_feder_problem} becomes equivalent to solving \eqref{eqn_the_central_problem} as the training sample size grows. First, let's examine the solution to the centralized problem \eqref{eqn_the_central_problem}.
\subsection{Learning the Centralized Problem}
Similarly to the agent problem \eqref{eqn_the_agent_problem} and the server problem \eqref{eqn_the_feder_problem}, the centralized problem \eqref{eqn_the_central_problem} is solved in the dual domain. In order to derive the dual problem, we first start by introducing the Lagrange multiplier $\bblam \in \reals^N_+$, associated with the inequality constraints. Formally, we introduce the Lagrangian 
\begin{equation}\label{eq_central_Lagrange}
\begin{aligned}
	\calL(\alpha,\bblam) &=
		\frac{1}{2} \Vert \alpha \Vert_{L_2}^2 + \gamma \Vert \alpha \Vert_{L_0} \\
		{} &+ \frac{1}{N} \sum_{n = 1}^N \bblam_{n} \ell(f(\bbx_n), y_n)
		\text{.}
\end{aligned}
\end{equation}
In a similar manner to the federated problem, the central learner obtains the dual function and the dual problem. The dual function is concave and therefore the dual problem is solved using gradient descent. The gradients are computed by evaluating the constraints at the variable $\alpha_d$, which minimises the Lagrangian
$
\alpha_d(\bbs,w) = \underset{\alpha\in L_2}{\text{argmin}} ~\calL(\alpha,\bblam).
$
The variable $\alpha_d$ which minimizes the Lagrangian \eqref{eq_central_Lagrange} has the following expression
\begin{equation}\label{eqn_alpha_d}
\alpha_d(\bbs, w;\bblam) = \begin{cases}  \bar{\alpha}_d(\bbs, w;\bblam)&  (\bar{\alpha}_d(\bbs, w;\bblam))^2 >2\gamma \\
0 & otherwise,
\end{cases} 
\end{equation}
for which
\begin{equation}
\bar{\alpha}_d(\bbs, w;\bblam) = \frac{1}{N} \sum_n \bblam_{i,n} y_n k(\bbs, \bbx_n,w).
\end{equation}
Using \eqref{eqn_alpha_d}, we can form a closed form expression for the dual function as the quadratic function, given the measure $m(\calX, \calW) = \int \indicator[\alpha_l(\bbs,w) \neq 0] d\bbs \,dw$ 
\begin{equation}\label{eq_dual}
\begin{aligned}
g(\bblam) = &-0.5 \bblam^\top \bbQ \bblam
+ \frac{1}{N}\bblam^\top (\textbf{1} - \bbepsilon) + m(\ccalX, \ccalW), 
\end{aligned}
\end{equation}
with $\bbQ$ being a positive definite matrix for which
\begin{equation}
\bbQ_{nm} = \int_{\calC} \frac{1}{N^2} y_n y_m k(\bbx_n,\bbs;w)k(\bbx_m,\bbs;w)  d\bbs dw,
\end{equation}
where $\calC = \{ (\bbs,w)~|~ \alpha_d(\bbs,w) \neq 0\}$.  
\subsection{Critical Samples}
In section \ref{sec:Algorithm}, we have claimed that the critical samples are determined by the values of the optimal dual variable. Particularly, given a set of samples, only the sample points which contribute to the classification model are considered critical. 
The following proposition shows that these critical points are not just particular to this training set but to the classification problem in general.
\begin{proposition}\label{lem_zeros}
Let $\alpha^*$ be the optimal variable of \eqref{eqn_the_central_problem} trained on data set $\bbX$, $\alpha'^*$ be the optimal variable of \eqref{eqn_the_central_problem} trained on data set $\bbX' = \bbX \setminus \left\lbrace \bbx_n\right\rbrace $ and $\hat{y}_n = \int \alpha'^*(\bbs,w) k(\bbx_n,\bbs,w) d\bbs dw$.
The dual optimal variable associated with the nth sample, $\lambda_n^* = 0$ if and only if  $1 - \epsilon-y_n  \hat{y}_n < 0$  and the solutions to the data $\bbX$ and the data $\bbX'$ are equal. 
\end{proposition}
\begin{proof}
See Appendix \ref{A:Fund}
\end{proof}
This proposition implies that if the federated learner \eqref{eqn_the_feder_problem} and the centralized learner \eqref{eqn_the_central_problem} agree on the critical samples then solving the two problems is equivalent. Furthermore, it is sufficient for the agent learner \eqref{eqn_the_agent_problem} to agree with the centralized learner despite only sampling from a subspace of $\calX$. Next we will argue that this is in fact the case as the sample size grows.

We consider the case in which the subspaces sampled by the agents are not distinct, i.e., there exists at least one pair $i,j$ such that $\calX_i \cap \calX_j \neq \emptyset$. 
If all subspaces are disjoint the problem becomes trivial. In this case, there is no need to form a global model because the agents do not gain useful information from other agents. Given a new sample, its classification can be done by simply finding the space to which it belongs and using the model of the respective agent. It should be noted that the problem of identifying the subspace is not trivial, yet in a federated learning setting a new sample generally, belongs to the subspace of the agent that has collected it. Similarly, in the case in which agents sample over the same space, i.e., $\calX_i = \calX_j, \forall~ i,j$ the need for sharing data across agents disappears. As the agents collect more data their models will converge. 
We are therefore, interested in the case for which there exist at least one  pair of agents $A_i, A_j$ such that $\ccalX_i \cap \ccalX_j \neq \emptyset$ and $\calX_i \neq \calX_j$. 
We make the following hypothesis about the kernels centered at points belonging exclusively to one subspace 

\begin{hyp}\label{h:overlap}
The overlap between any two partitions is large enough such that there exists a small $\xi>0$ such that for all $\bbx_i \in \{\calX_i \setminus \calX_j \}$ and $\bbs \in \{\calX_j \setminus \calX_i \} $ and  $w\in \calW$ 
\begin{equation}\label{assum_ovl}
k(\bbx_i, \bbs;w) \leq \xi.
\end{equation}
\end{hyp}
This hypothesis implies that samples which uniquely belong to a subspace do not affect the models of other subspaces in the non-overlapping regions. Indeed, recall that the function $\alpha(\bbs,w)$ is a weighted sum of the kernels centered at the sample points, therefore a point located in a non overlapping part of a subspace the weighted sum of the kernels outside that partition will be small and not contribute significantly to the value of $\alpha$.
Additionally, we make the following assumption about the choice of $\gamma$
\begin{hyp}\label{h:big_gamma}
Let ~ $\calC = \{ (\bbs,w)~|~ \alpha^*(\bbs,w) \neq 0\}$ be the support of the optimal value $\alpha^*(\bbs,w)$ of \eqref{eqn_the_agent_problem}. We choose the variable $\gamma$ that leads to $\calC$ being rich enough such that there exists a $\mu>0$ for which
\begin{equation}
\bbQ = \int_\calC \bbq(\bbs,w)\bbq^\top(\bbs,w) ds\,dw \succeq \mu \bbI
\end{equation}
where ~the variable $\mu$ represents the smallest eigenvalue of the matrix $\bbQ$, and ~$\bbq_n(\bbs,w) = (1/N) y_n k(\bbx_n,\bbs;w)$.
\end{hyp}
Notice that this hypothesis relies on the choice of $\gamma$. In theory, the choice of $\gamma$ does not cause the measure of $\calC$ to go to zero. In practice, however, the choice of $\gamma$ becomes more important (see Remark \ref{rem_gamma}). 
Furthermore, the hypothesis suggests that the matrix $Q$ is positive definite. As a consequence, the dual function is strongly concave near the optimal  $\bblam^*$ and we can formulate the following lemma.
\begin{lemma}\label{strong_concav}
The dual function $g(\bblam)$ for the problem in \eqref{eqn_the_agent_problem} is strongly concave near the optimal value, $\bblam^*$. The strong concavity parameter $\mu$ as defined in Hypothesis \ref{h:big_gamma} such that
\begin{equation}
-g(\bblam) \geq -g(\bblam^*) - \nabla g(\bblam^*)^T(\bblam-\bblam^*) + \frac{\mu}{2} \Vert \bblam-\bblam^* \Vert^2
\end{equation}
and $\mu$ corresponds to the smallest eigenvalue of $\bbQ$.
\end{lemma}

\begin{proof}
First recall the definition of the dual function:
\begin{equation}\label{dual_red}
g(\bblam) = -\frac{1}{2} \bblam^\top \bbQ \bblam + \bblam^\top (\textbf{1} - \bbepsilon) + m(\calX, \calW)
\end{equation}
There must exist a variable $\delta > 0$ such that
\begin{equation}
g_i(\lambda^* + \delta) < g_i(\lambda^*)
\end{equation}
with $\delta$ close to zero. We will show that there exists a $\mu > 0$ such that
\begin{equation}\label{strong_convex}
g(\lambda^*)-g(\lambda^* + \delta) +\nabla g(\lambda^*)^\top (\delta) \geq  \frac{\mu}{2} \Vert \delta \Vert^2
\end{equation} 
We can calculate the value on the right side of the inequality we assume that the support of the matrix $\bbQ$ is approximately equal for $\bblam^*$ and $\bblam^* + \delta$. 
\begin{equation}
\begin{aligned}
&g(\lambda^*) -g(\lambda^* + \delta)  + g(\lambda^*)^\top (\delta)  = \\
&-\frac{1}{2} \sum_i \sum_j \lambda_i^* \lambda_j^* \bbQ_{ij} + \sum_i \lambda_i^*(1-\epsilon_i) + m(\calS, \calW)  \\
&+ \frac{1}{2} \sum_i \sum_j \lambda_i^* \lambda_j^* \bbQ_{ij} + \sum_i \sum_j \lambda_i^* \delta_j \bbQ_{ij} + \frac{1}{2} \sum_i \sum_j \delta_i \delta_j \bbQ_{ij} \\
&-\sum_i \lambda_i^*(1-\epsilon_i) - \sum_i \delta_i(1-\epsilon_i) - m(\calS,\calW) \\
&-\sum_i \sum_j \lambda_i^* \delta_j \bbQ_{ij} + \sum_i \delta_i(1-\epsilon_i) = \\
&\frac{1}{2}  \bbdelta^\top \int_\calC \bbq(s,w) \bbq^\top(s,w)ds dw ~\bbdelta \geq \frac{\mu}{2} \Vert \bbdelta \Vert^2
\end{aligned}
\end{equation}
This proves that the dual function is strongly concave near the optimal value.
\end{proof}
Notice that as Hypothesis \ref{h:big_gamma} and Lemma \ref{strong_concav} apply not only to problem \eqref{eqn_the_agent_problem}, they also apply to problems \eqref{eqn_the_central_problem} and \eqref{eqn_the_feder_problem}.
Given two hypothesis, we can state that the solutions of \eqref{eqn_the_feder_problem} and \eqref{eqn_the_central_problem} converge to each other.
\begin{theorem}\label{th_part}
Let $\alpha_C^*$ and $\alpha_F^*$ be the solution to the problem \eqref{eqn_the_central_problem} and \eqref{eqn_the_feder_problem} respectively.
Given that hypotheses \ref{h:overlap} and \ref{h:big_gamma} hold, the two solutions converge, as the sample size grows 
\begin{equation}
\lim_{N\rightarrow \infty} \vert \alpha_F^*(\bbs,w) - \alpha_C^*(\bbs,w) \vert \rightarrow 0.
\end{equation}
\end{theorem}
\begin{proof}
See Appendix \ref{A:Theoreml}.
\end{proof}
In order to understand the proof of this theorem it is necessary to examine the two cases: the case of agents sampling the same space and the case of agents sampling disjoint spaces.
Consider the case in which the two agents are observing completely separate spaces. We assume the two spaces are far apart such that the kernel value for two points in the separate spaces takes on a small value.
\begin{hyp}\label{h:distance}
Let $\calX_i$ and $\calX_j$ be two supspaces of $\calX$ which do not overlap ($\calX_i \cap \calX_j = \emptyset$) and $w\in \calW$. Then for $\xi$ from Hypothesis \ref{h:overlap} the following holds
\begin{equation}\label{eq_assum}
k(\bbx_i, \bbx_j;w) < \xi, \, \forall \, w \in \calW, \bbx_i \in \calX_i, \bbx_j \in \calX_j.
\end{equation}
\end{hyp}
Notice that this hypothesis is similar to Hypothesis \ref{h:overlap} and implies that samples from one subspace do not affect the solution of another subspace. From this assumption we can formulate the following lemma about the global dual function with respect to the local dual functions.

\begin{lemma}\label{lemma_distinct}
Given a group of agents which sample separate spaces as dictated by hypothesis \ref{h:distance}, let $g$ be the global dual function and $g_i$ be the agent dual function for agent $i$. Then
\begin{equation}\label{eqn_dual_distinct}
\vert g(\bblam) - \sum_i g_i(\bblam_i) \vert \leq \frac{2 \xi m L}{N^2}, 
\end{equation} 
for any $\bblam = N [\bblam_1^\top /{N_1}, \dots, \bblam_K^\top/{N_K}]^\top$, where $L = \frac{N^2}{N_1,N_2}\sum_i \sum_{j \neq i} \bblam_i^T \bbJ \bblam_j$ and $m$ is the measure of the support of the function $\alpha_d$. $\bbJ$ is an all-ones matrix. 
\end{lemma}

\begin{proof}
See Appendix \ref{A:Lemma_d}
\end{proof}
Notice that the values of the dual variables are weighted by the number of samples which means $\bblam/N = [\bblam_1^\top /{N_1}, \dots, \bblam_K^\top/{N_K}]^\top$. This causes the primal variables to be at the same scale despite being a sum of kernels weighted by the number of samples. Therefore, we can establish the following theorem
\begin{theorem}\label{theorem:separate}
Let $\alpha_C^*$ and $\alpha_i^*$ be the solution to the problem \eqref{eqn_the_central_problem} and \eqref{eqn_the_agent_problem} respectively.
Given that hypotheses \ref{h:distance} and \ref{h:big_gamma} hold, the two solutions converge, as the sample size grows 
\begin{equation}
 \vert \alpha_C^*(\bbs,w) - \sum_i\alpha_i^*(\bbs,w) \vert  \leq \frac{2\sqrt{2 \xi m L}}{N \sqrt{\mu N}}.
\end{equation}
\end{theorem}
\begin{proof}
See Appendix \ref{A:theorem_sep}
\end{proof}
This theorem establishes the relationship between the primal variables of \eqref{eqn_the_agent_problem} and \eqref{eqn_the_central_problem} over non-overlapping areas.
In the case in which agents observe the same space we formulate the following theorem.
\begin{theorem}\label{th_overlap}
Let $\alpha_i^*$ and $\alpha_j^*$ be the optimal variables to problem \eqref{eqn_the_agent_problem} solved by agent $i$ and $j$ respectively. The agents observe samples independently over the same distribution. Further, let $M \geq \Vert f \Vert^2$, $c$ be the Berry-Essen theorem constant,  $\rho = \mathbb{E}_\bbx\left[ \left| \lambda(\bbx) y_\bbx k(\bbx,\bbs;w) \right|^3\right]$ and $\sigma^2 = \mathbb{E}_\bbx\left[ \left| \lambda(\bbx) y_\bbx k(\bbx,\bbs;w) \right|^2\right]$. Let $\mu = min(\mu_i, \mu_j)$ for which $\mu_i$ and $\mu_j$ are the minimum eigenvalue of $\bbQ_i$ and $\bbQ_j$ respectively. Then the difference between the solutions computed by the two agents has the following bound
\begin{equation}
\vert \alpha_i(\bbs,w) - \alpha_j(\bbs,w) \vert \leq   2 \left(2\sqrt{\frac{M}{\mu N^{1.5}}} + \frac{c \rho}{\sigma^3 \sqrt{N}} \right),
\end{equation}
where $N = min(N_i, N_j)$ is the minimum of the two samples sizes.
\end{theorem}
\begin{proof}
See Appendix \ref{A:th_overlap}
\end{proof}
Because agents sample the same space, as the sample size grows the solutions from two agents converge. In fact, if the solutions of two agents are reciprocally feasible, then they are equal (Lemma \ref{lem_dif}). The centralized learner can be viewed in this case as an agent which collects more samples therefore the solution of an agent converges to that of the centralized learner as well.

\begin{lemma}\label{lem_dif}
Given two problems as in \eqref{eqn_the_agent_problem} with different sample sets, let $P_1$ and $P_2$ be the solutions to these problems  If the two solutions $P_1$ and $P_2$ are reciprocally feasible, that is if the optimal variable $\alpha_1^*$ is feasible to the second problem and vice versa then the two solutions are equal:

\begin{equation}
P_1 = P_2.
\end{equation}
\end{lemma}

\begin{proof}
Let us first notice that the objective function is independent of sample size and is therefore equal for both problems let us denote it as

\begin{equation}
f_0(\alpha) = \int \frac{1}{2} \alpha^2(s,w) + \gamma \mathbb{I}\left( \alpha(s,w) \neq 0 \right) ds \, dw.
\end{equation}
Suppose now that $P_2 > P_1$ this implies that there exists an $\alpha_1$ which is feasible in the second problem such that $f_0(\alpha_1) = P_1 < P_2$, however since $P_2 = \underset{\alpha}{\min} f_0(\alpha)$ by definition, it follows that $P_2 \leq P_1$. This implies that there exists an $\alpha_2$ which is feasible in the first problem such that  $f_0(\alpha_2) = P_2 < P_1$. However, $P_1$ is by definition optimal and therefore it must be that $P_1 = P_2$.  

\end{proof}

\begin{remark} \label{rem_gamma}
In Hypothesis \ref{h:big_gamma} we make an assumption about the parameter $\gamma$ being chosen such that the space $\calC$ is rich enough to assure strong concavity of the dual problem in Lemma \ref{lem_dif}. The strong concavity is guaranteed by the optimal $\alpha^*$ having a non-zero support measure. In theory, as long as the zero function, $f(x) = 0$ is not feasible because of the constraints, $alpha$ is guaranteed to have non-zero support. Increasing parameter $\gamma$ shrinks the support of $\alpha$ and consequently reduces the value of the strong concavity parameter $\mu$. Although, the dual problem still has strong concavity, the lower value of $\mu$ makes it more difficult to solve and therefore the dual problem algorithm requires more iterations to converge. 
\end{remark}

\section{Experimental Results}
\label{sec:Num}
In the previous section, we have proposed a federated learning model which (i) reduces the necessary communication complexity and (ii) converges to the omniscient unit solution as the sample size grows. In this section, we first show through a simulated signal that the solution of our \eqref{eqn_the_feder_problem} converges to that of \eqref{eqn_the_central_problem} with as the sample size grows. Then using an activity identification task, we demonstrate that our algorithm can significantly reduce the communication cost to the central unit without compromising the performance of the classification.
For the classification task, we use the family of RKHSs with Gaussian kernels
\begin{equation}
k(\bbx,\bbx^\prime) = \exp \left\lbrace \frac{-\Vert \bbx - \bbx^\prime \Vert^2 }{2 w^2} \right\rbrace.
\end{equation}
The width of the kernel is directly proportional to the hyper-parameter $w$.

To start, the effect of sample size on the generalization accuracy is examined on a simulated data set. To this end, we simulate a uniformly distributed signal and define the class membership for each sample as
\begin{equation}
y = \begin{cases} 1,  & (\bbx' - \bbc_i)^\top \bbA (\bbx' - \bbc_i) \leq r_i  ~\text{for any}~ i \\
-1 & otherwise,
\end{cases}
\end{equation}
for which $r_1 = 9$, $r_2 = 30$, $\bbc_1 = [3,0]^\top$, $\bbc_2 = [-10,6]^\top$ and $\bbA = [1,0;0,0.25]$. The space $\calX$ was divided into $9$ overlapping subspaces. Each agent collects data from only one subspace and forms its local model. There were $9$ subspaces from which the agents collect the data. The setup of subspaces and class labels can be seen in Figure \ref{fig:setup}.
After samples are assigned to a class, Gaussian noise is added to the samples $\bbx = \bbx'+ \xi$, where $\xi \ \in  \calN(0,0.2)$ in order to create noisy samples. A separate testing set of  $1000$ samples is created for the evaluation of the learner. 

The performance of the federated learner \eqref{eqn_the_feder_problem} is compared to that of the centralized learner \eqref{eqn_the_central_problem}. Both methods used $\gamma = 25$, $\epsilon = 10^{-2}$ and a learning rate of $\eta = 0.1$. 
Figure \ref{fig:sample} compares the generalization accuracy of the two learners over training sample sizes ranging from $90$ to $900$. The average generalization accuracy was calculated over 100 repetitions. When the sample size is below $400$ the federated classification learner has a better generalization accuracy. This is most likely due to the agents being able to learn a simpler problem despite having a small training set.
 As the sample size grows the two solutions converge which is reflected by the generalization accuracy converging. 
\begin{figure}
\centering
\includegraphics[width=0.8\linewidth]{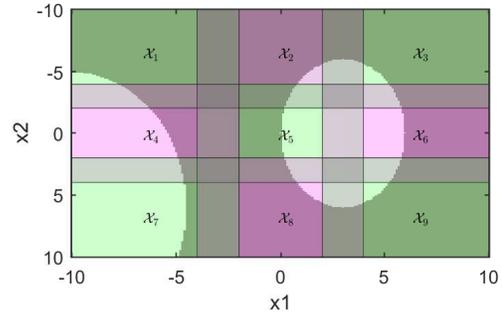}
\caption{The simulated space $\calX$. The subspaces sampled by each agent are colored either purple or green with the gray spaces being sampled by multiple agents. The class membership is determined by the brightness: the bright areas belong to class +1, and the darker areas belong to class -1.}
\label{fig:setup}
\end{figure}
\begin{figure}[tb]
  \centering
  \centerline{\includegraphics[width=0.8\linewidth]{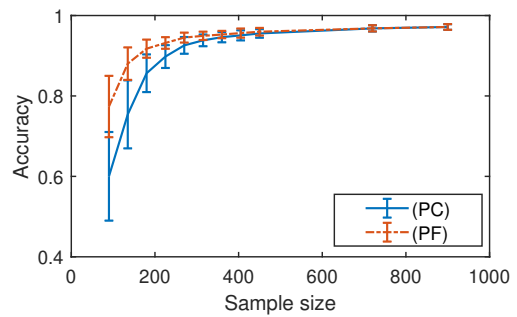}}
\caption{The average accuracy taken over 100 repetitions of randomized sampling of the federated learner \eqref{eqn_the_feder_problem} and that of the centralized learner \eqref{eqn_the_central_problem} as a function of sample size.}
\label{fig:sample}
\end{figure}

\subsection{Task Classification}
We further evaluate our methods using biometric data \cite{weiss2019smartphone} containing measurements from the accelerometer sensor from a smartphone. The study contained measurements from participants while performing various tasks, such as jogging, walking, writing, and typing.

The smartphone of each participant is considered an agent collecting data over its distribution. The agents collect data over different spaces since people don't perform the same activity the same way, e.g. some people walk faster, some people type slower, some write cursive, etc. Similarly, since all participants perform the same task the spaces from which the agents are collecting the data should not be distinct.

We examine the effect of the number of agents on the performance of our federated learner for the classification of running versus jogging. Agents are selected randomly from our training set and the data from each agent is randomly split into a training and a testing set. The time series from the phone's accelerometer is divided into $5$ second intervals, with each interval considered a sample. The average value is taken such that each sample contains three features. Then we train our federated learner and the centralized learner and compare the accuracy on the test set. Both learners use the following parameters: $\gamma = 100$, $\eta = 0.1$, $T = 1000$ and $\epsilon = 0.5$. This procedure was repeated $100$ times in order to obtain average performance. The federated learner \eqref{eqn_the_feder_problem} and the centralized learner \eqref{eqn_the_central_problem} have comparable average accuracy. When the number of agents is increased to $51$ agents the average accuracy of the federated learner is $77.35\%$ and the average accuracy of the centralized learner is $75.29\%$.
\begin{figure}
\centering
\includegraphics[width=0.8\linewidth]{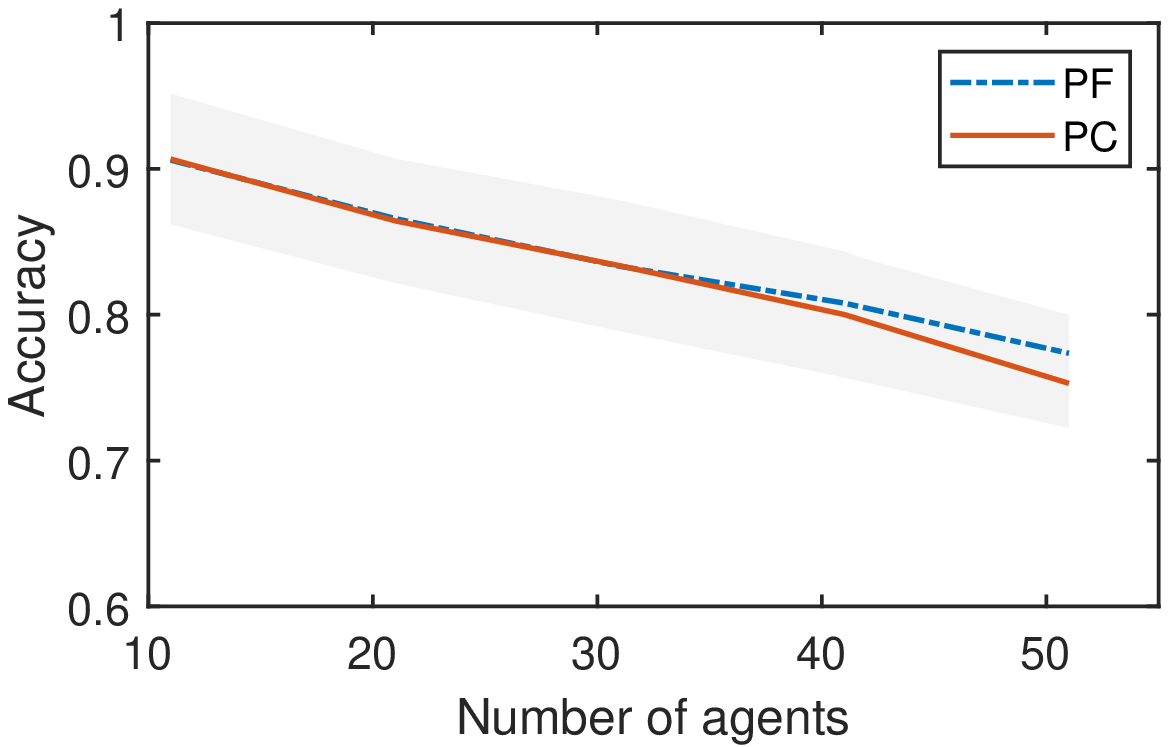}
\caption{The performance of the the federated learner and the centralized learner as a function of the number of agents.}
\label{fig:NumAgents}
\end{figure}

The effect of the number of agents is evaluated on a second task: typing and writing. The learners use the following parameters : $\gamma = 150$, $\eta = 0.1$, $T = 500$ and $\epsilon = 0.5$. In this case, the average accuracy decreases as the number of agents increases for both the federated learner and the centralized learner. The performance of the learners could potentially be improved by increasing the parameter $\epsilon$ for a larger number of agents.

\begin{figure}
\centering
\includegraphics[width=0.8\linewidth]{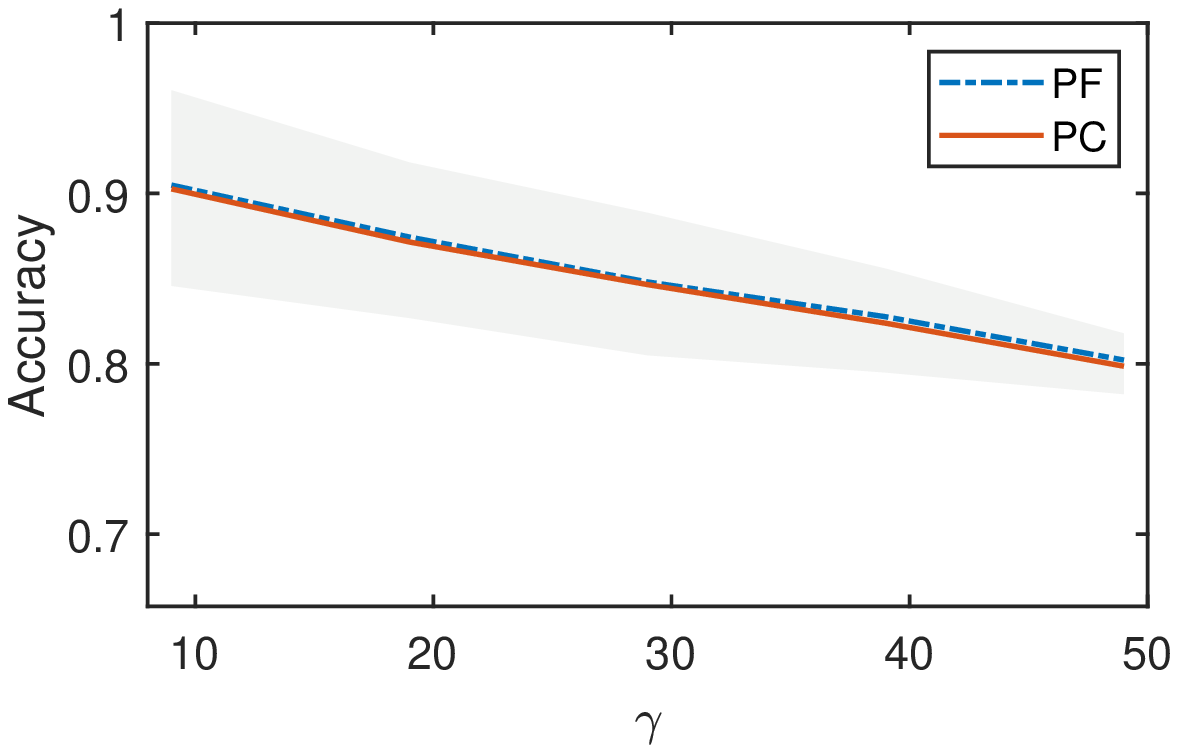}
\caption{The performance of the the federated learner and the centralized learner as a function of the number of agents.}
\label{fig:NumAgents_WrvsT}
\end{figure}

Next, we examine the effect of the sparsity parameter $\gamma$ on the performance of the learners. Data from $10$ participants is used to distinguish between the activities of walking and jogging. The regularizing parameter $\gamma$ which controls the complexity of the representation was varied to observe the effects on three metrics: accuracy of classification, the cost of communication, and the cost of the representation. The accuracy is measured as the percentage of correctly classified tasks. The communication cost is measured as the average number of samples that need to be transmitted over the channel. The representation cost is determined by the number of kernels used in the resulting global model.

The features were extracted from averaging over $5$ second intervals. The data is randomly split to create a training set and a test set. The accuracy is evaluated on the test set. The parameters used by both agents are:  $\eta = 0.1$, $T = 500$ and $\epsilon = 0.5$. The federated learner and the centralized learner are trained over $100$ random splits and the resulting accuracy, representation cost, and communication cost is averaged over those repetitions. The average accuracy does not change for either learner with respect to the sparsity parameter and both learners have similar performance Figure \ref{fig:class} (a). This is expected since the algorithm can produce representations of varying complexity without sacrificing performance. The average representation cost of the global model is inversely proportional to the sparsity parameter. Both learners achieve similar representation costs as can be seen in Figure \ref{fig:class} (b). The communication cost of the federated learner is directly proportional to the sparsity parameter. This is expected since low complexity representations for $\alpha(\bbs,w)$ require more intricate kernel functions and therefore more samples. Therefore, there exists a trade-off between the complexity of the representation of the global model and the communication cost of sending data over the network. If sparsity of the global model is not a concern the federated learner can achieve a communication cost that is $40\%$ of the communication cost of the centralized learner (Figure \ref{fig:class} (c)).

\begin{figure}[htb]     
      \begin{minipage}[b]{\linewidth}
      \centering
              \includegraphics[width=0.8\linewidth]{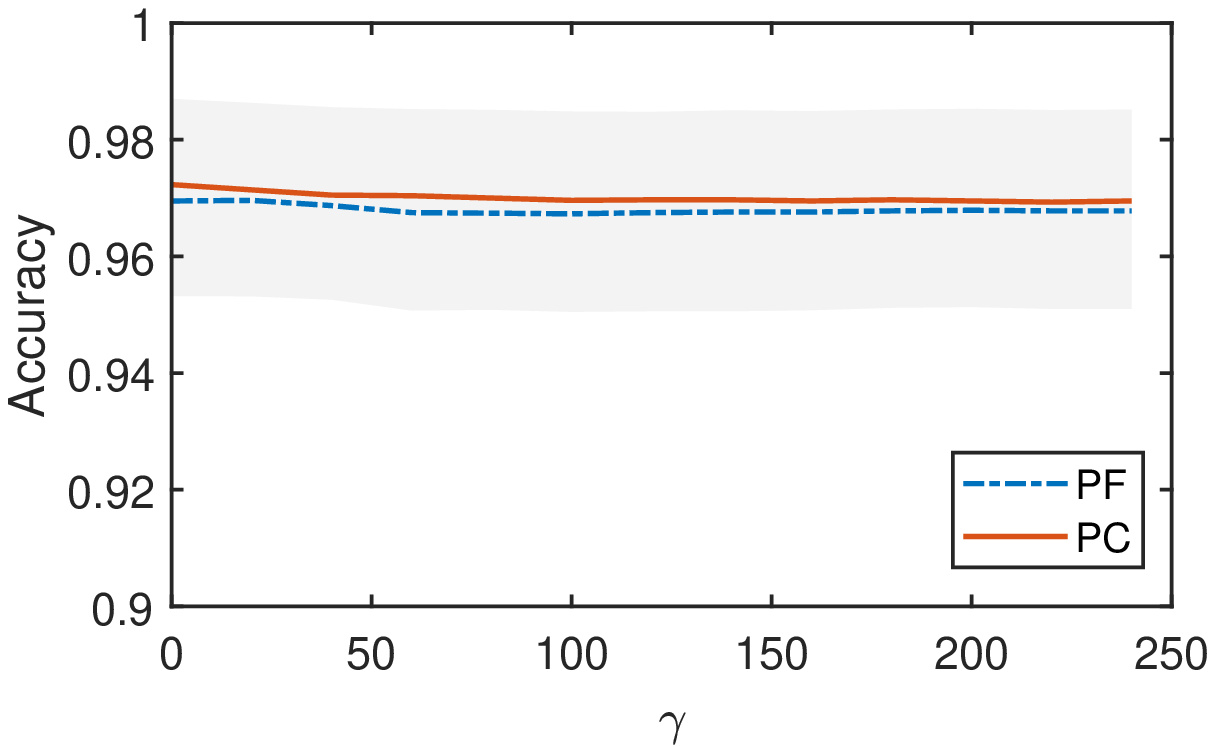}
              \centerline{(a)}
      \end{minipage}
      \begin{minipage}[b]{\linewidth}
      \centering
              \includegraphics[width=0.8\linewidth]{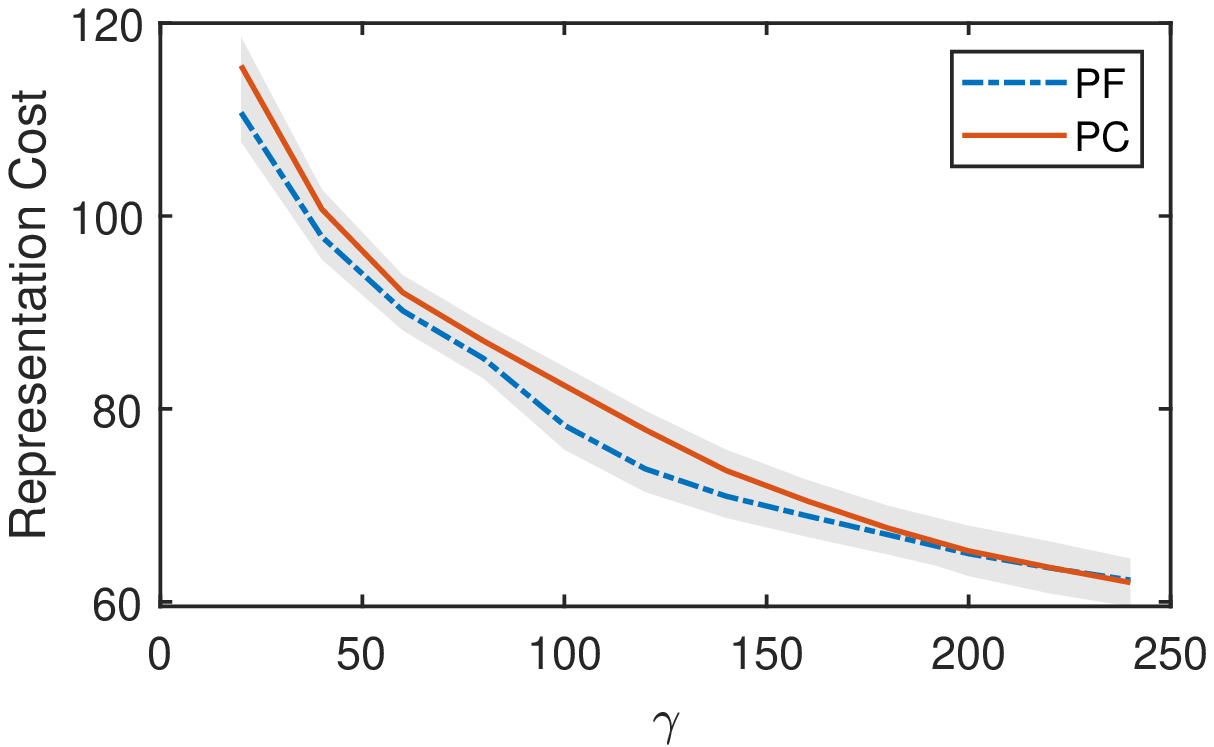}
               \centerline{(b)}
      \end{minipage}
      \begin{minipage}[b]{\linewidth}
      \centering
              \includegraphics[width=0.8\linewidth]{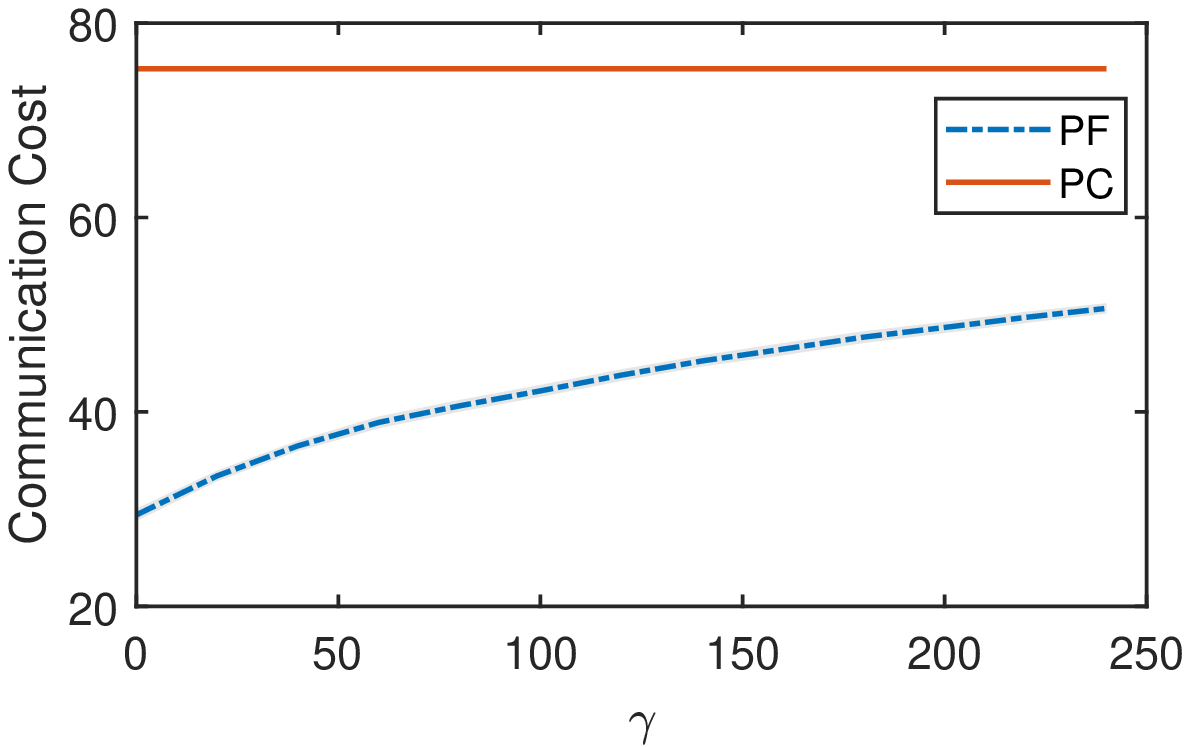}
               \centerline{(c)}
      \end{minipage}
      
      \caption{Classification of walking versus running using the federated learner and the centralized learner. (a) The accuracy of the federated classification learner and the centralized learner as a function of the sparsity parameter. (b) The representation cost of both learners as a function of the sparsity parameter. (c) The communication cost of transmitting data to the central unit for the federated learner and the centralized learner as a function of the sparsity parameter.}
\label{fig:class}
  \end{figure}

 We further validated our method by examining the problem of classification of writing and typing with data acquired from the phone accelerometer \cite{weiss2019smartphone}. The features are obtained by averaging over a $5$ second time window. The performance of our federated learner was compared to that of the centralized learner on the three metrics: accuracy, communication cost, and representation cost. The data from the agents was split randomly using $100$ repetitions. The parameters used by both agents are:  $\eta = 0.1$, $T = 500$ and $\epsilon = 0.5$. The sparsity parameter $\gamma$ was varied between $0$ and $240$.

The federated learner \eqref{eqn_the_feder_problem} and the centralized learner \eqref{eqn_the_central_problem} have similar accuracy and their performance is not affected by the sparsity parameter. This implies that the functions needed to represent the class difference are sufficiently sparse. The sparsity parameter controls the complexity of the representation which can be seen in the representation cost (Figure \ref{fig:WrvT}, (b)). Both learners achieve similar representation costs. The advantage of the federated learner comes from reducing the communication cost Figure \ref{fig:WrvT} (c). When sparsity is not required the federated learner achieves a reduction of $64\%$ in communication cost.

  \begin{figure}[htb]
      \begin{minipage}[c]{\linewidth}
      \centering
              \includegraphics[width=0.8\linewidth]{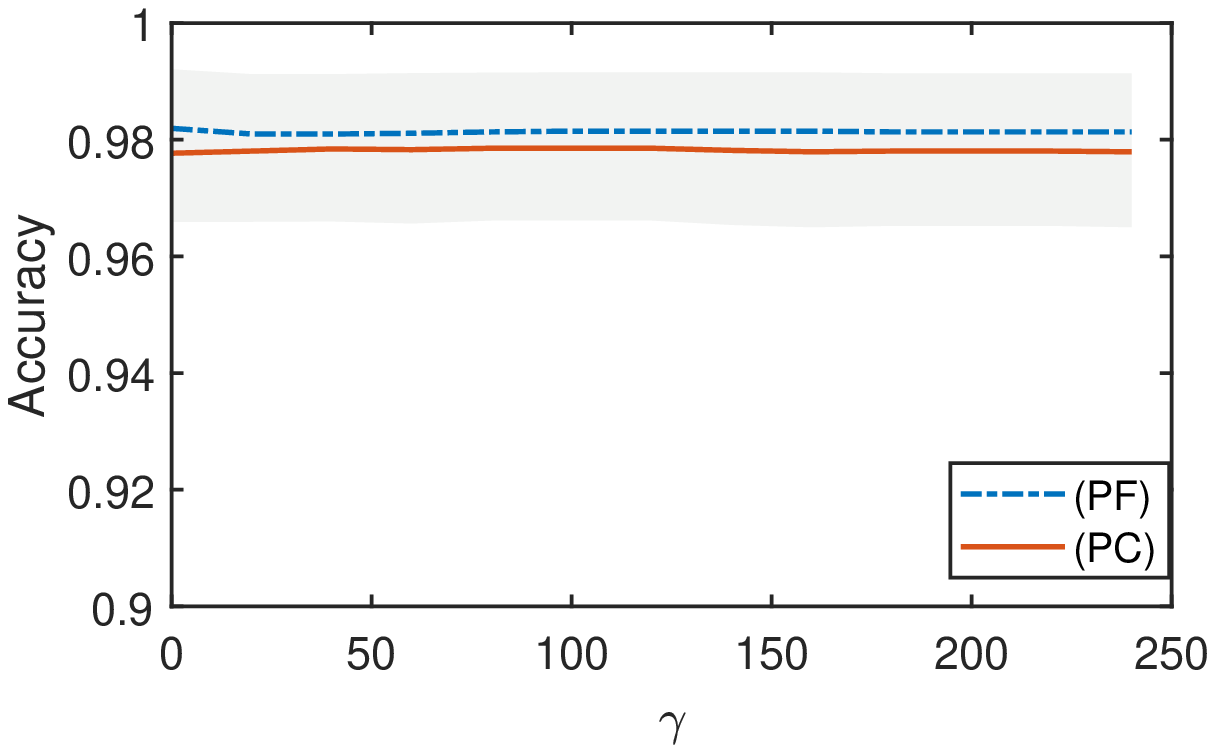}
              \centerline{(a)}
      \end{minipage}
      \begin{minipage}[c]{\linewidth}
      \centering
              \includegraphics[width=0.8\linewidth]{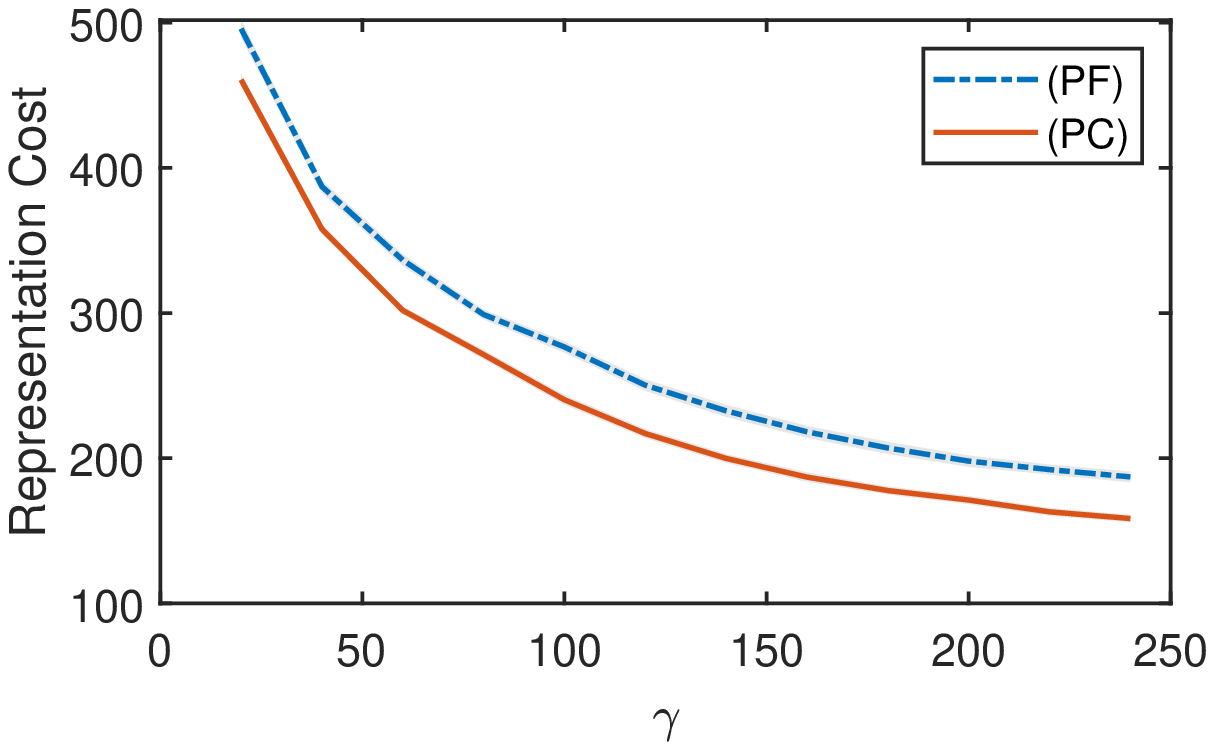}
               \centerline{(b)}
      \end{minipage}
      \begin{minipage}[c]{\linewidth}
      \centering
              \includegraphics[width=0.8\linewidth]{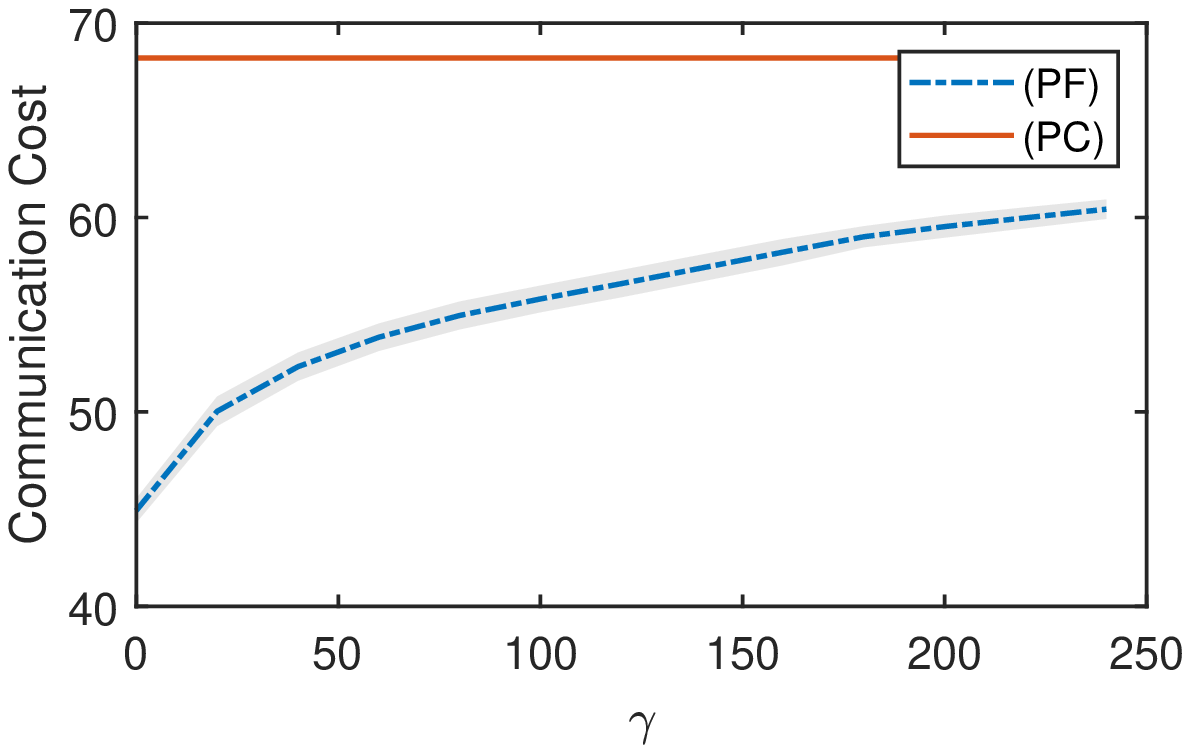}
               \centerline{(c)}
      \end{minipage}
      \caption{Classification of writing versus typing using the federated learner and the centralized learner. (a) The accuracy of the federated classification learner and the centralized learner as a function of the sparsity parameter. (b) The representation cost of both learners as a function of the sparsity parameter. (c) The communication cost of transmitting data to the central unit for the federated learner and the centralized learner as a function of the sparsity parameter.}
\label{fig:WrvT}
  \end{figure}

\section{Conclusion}
This paper introduced a method for federated classification using low complexity RKHS representations. This was achieved by first introducing a method for traditional learning, which obtains both a sparse representation and identifies the critical samples for the classification problem. By leveraging the ability to detect the critical samples to the classification problem our federated learner is able to reduce traffic over the network and send less information. The federated classification method was shown to converge to a traditional learning method in which the learner has access to the entire data set as the sample size grows. The federated learner was used in a task recognition problem for which the data was collected from users' phones. In our numerical experiments, our method significantly reduced the communication cost while maintaining a similar accuracy and complexity of representation to the centralized learner.

\appendices

\section{Proof of Lemma \ref{lem_zeros}}\label{A:Fund}
\begin{proof}
For this proof we will establish the following notation in order to make the proof easier to read:
\begin{equation}
\begin{aligned}
k_n = k(x_n,s;w);
\end{aligned}
\end{equation}

The first term of dual function can be rewritten as:

\begin{equation}\label{eqn_1lem_zeros}
\begin{aligned}
\bblam^\top \bbQ \bblam =  \frac{1}{N^2}\int \sum_n \sum_m \lambda_n \lambda_m k_n k_m y_n y_m d\bbs \, dw \\
= \frac{1}{N^2}\int \left( \lambda_n y_n k_n + \sum_{m \neq n} \lambda_m y_m k_m \right)^2 ds \, dw \\
= \bblam'^\top \bbQ' \bblam' + \frac{1}{N^2} \left( 2 \lambda_n y_n k_n \sum_{m \neq n} \lambda_m y_m k_m + \lambda_n^2 k_n^2 \right)
\end{aligned}
\end{equation}
where $\bblam'$ and $\bbQ'$ are the variables $\bblam$ without the $n^{th}$ element and $\bbQ$ without the $n^{th}$ row and column respectively. Using \eqref{eqn_1lem_zeros} we can rewrite the dual function:

\begin{equation}\label{eqn_zeros}
\begin{aligned}
g(\bblam) = -0.5 \bblam'^\top \bbQ' \bblam' + \frac{1}{N}\bblam'^\top (1-\epsilon) + \gamma m(\ccalX, \ccalW) \\
 +\frac{1}{N} \lambda_n  \left(1 -\epsilon - y_n \frac{1}{N}\int \sum_{m \neq n} \lambda_m y_m k_m k_n ds \, dw \right) \\
 - \frac{0.5}{N^2}\int \lambda_n^2 k_n^2 ds \, dw
 \end{aligned}
\end{equation}
Notice that $(1/N)\int \sum_{i \neq n} \lambda_i k_i k_n d\bbs dw$ evaluated at the optimal $\bblam'$ is precisely $\hat{y}_n$ considering $\bblam'_m = \bblam_m(N-1)/N)$ for all $m$ and $g(\bblam| \bblam_n = 0) = g((N/(N-1) \bblam')$. Since, $\bblam_n = 0$ it follows from complementary slackness that $1-\epsilon-y_n \hat{y}_n < 0$. Therefore, it follows that the optimal values for the two dual functions are equal if $\lambda_n^* = 0$. Moreover, the optimal primal variables are equal, i.e., $\alpha^*(\bbs,w) = \alpha'^*(\bbs,w)$. This concludes the first part of the proof.

Next, we will show that if a model that is optimal for $X'$ and that has the property $1-\epsilon - y_n \hat{y}_n <0$ for a new sample $\bbx_n$, the optimal dual variable corresponding to that point for the model trained on the set $X = X' \cup \lbrace x_n \rbrace$ has value $\bblam_n^* = 0$.
Equation \eqref{eqn_zeros}  implies that optimizing for the variable $\bblam_n$, given the solution to the model using $X'$ results in $\lambda_n = 0$. This value maximizes the dual function $g(\bm{\lambda})$. It is necessary to prove that there is not a value for $\bm{\lambda}$ different from $\bm{\lambda}'^*$ for which $g(\bm{\lambda}'^*) < g(\bm{\lambda}^*)$.
Since $1 -\epsilon -y_n \hat{y}_n < 0$, the optimal $\alpha'^*$ is feasible for the model which uses $x_n$ as a sample as well and it has not been proven yet to be optimal for the full set we can say $P'^* \geq P^*$. However, since we have strong duality it is also true that 

\begin{equation}
g(\bm{\lambda}'^*) = P'^* \geq P^* = g(\bm{\lambda}^*)
\end{equation}
Since $g(\bm{\lambda}^*)$ is the maximum over $\bm{\lambda}$ it follows that  $g(\bm{\lambda}'^*) = g(\bm{\lambda}^*)$, which implies that $\lambda_n =0$. This concludes the proof.

\end{proof}

\section{Proof of Theorem~\ref{th_part}}
\label{A:Theoreml}
\begin{proof}
Given two data sets $\bbX_i$ and $\bbX_j$ drawn over partitions of the space $\calX = \calX_i \cup \calX_j$, let $\alpha^*(\bbs,w)$ be the solution to the problem \eqref{eqn_the_central_problem} given $[\bbX_i,\bbX_j]$ as a training set and, $\alpha_{(i)}^*(\bbs,w)$ and $\alpha_{(j)}^*(\bbs,w)$ be the solution to the problem \eqref{eqn_the_agent_problem} trained on $\bbX_i$ and $\bbX_j$ respectively.
Additionally, let the overlap be large enough that Hypothesis \ref{h:distance}  holds for the non-overlapping spaces.
Let $\calX_o = \calX_i \cap \calX_j$ be the overlapping space and $\calX_i' = \calX_i \setminus \calX_o$, $\calX_j' = \calX_j \setminus \calX_o$.
Then we can write the $\alpha(\bbs,w)$ as a sum of functions which are nonzero only over one space where $\alpha_i(\bbs,w) = 0 ~\forall ~\bbs \notin \calX_i'$, $\alpha_j(\bbs,w) = 0 ~ \forall ~ \bbs \notin \calX_j'$ and  $\alpha_o(\bbs,w) = 0 ~\forall~ \bbs \notin \calX_o$
\begin{equation}
\alpha(\bbs,w) = \alpha_i(\bbs,w) + \alpha_j(\bbs,w) + \alpha_o(\bbs,w)\
\end{equation}
it follows from Theorem \ref{th_overlap} that the each $\alpha_{(j),o}(\bbs,w)$ and $\alpha_{(i)o}$ converge to each other as the number of samples grows 
\begin{equation}
\vert \alpha_{(j)o}^* - \alpha_{(i)o}^*\vert \leq 2 \left(2\sqrt{\frac{M}{\mu N^{1.5}}} + \frac{c \rho}{\sigma^3 \sqrt{N}} \right)
\end{equation}
As the sample size grows the functions $f_i$ and $f_j$ over the overlapping space $\calX_o$ converge and therefore if for a point $\bbx \in \calX_o$, if $1- \epsilon_\bbx -yf_i(\bbx) <0$ then it must also hold that $1 - \epsilon_\bbx - y f_j(\bbx)<0$. Because the agents sample over the overlapping area, as the sample size grows and the agents agree on the critical samples, they will also agree with the centralized learner on the critical samples. Then according to Lemma \ref{lem_zeros} the solution of \eqref{eqn_the_feder_problem} and \eqref{eqn_the_central_problem} will converge over $\calX_o$. Although this was illustrated for two agents, the proof holds for any number of agents.

Over the spaces which do not overlap, consider \eqref{eqn_the_agent_problem} trained on $\bbX_i$ and \eqref{eqn_the_central_problem} trained on $\bbX$ and let  $\alpha_i^*$ and $\alpha_{(i)i}^*$ be their respective optimal values. Then for $ \bbs \in \calX_i'$ we can establish the following
\begin{equation}
    \begin{aligned}
    &\vert \alpha_{i}(s,w)^* - \alpha_{(i)i}^*(s,w)\vert \leq \\
&\left\vert \alpha^*(\bbs,w) - \sum_i \alpha_{(i)}^*(\bbs,w) \right\vert  + \vert \sum_{j \neq i} \alpha_{(j)}(s,w) \vert \\ 
&\leq \frac{2\sqrt{2 \xi m L}}{N \sqrt{\mu N}} + \sum_{j \neq i} \vert\alpha_{(j)}(s,w) \vert\\
&\leq \frac{2\sqrt{2 \xi m L}}{N \sqrt{\mu N}} + \xi \sum_{j\neq i} \frac{\Vert \bblam_j\Vert_1}{N_j}.
    \end{aligned}
\end{equation}
Notice that $\alpha_j(\bbs,w)$ has little effect on the value of $f(\bbx)$ for $\bbx \in \bbX_i$. As the sample size grows, if $1- \epsilon_\bbx yf_i(\bbx) <0$ then it must also hold that $1- \epsilon_\bbx yf(\bbx) <0$, where $f_i$ and $f$ are the solutions found by agent i and the centralized learner respectively. Then according to Lemma \ref{lem_zeros} as the sample size grows the agents and the centralized learner agree on the critical samples and the federated learner \eqref{eqn_the_feder_problem} and the centralized learner \eqref{eqn_the_central_problem} solve more similar problems.  
\end{proof}

\section{Proof of Lemma~\ref{lemma_distinct}}\label{A:Lemma_d}
\begin{proof}
We first show that it is true for two agents and then expand it for multiple agents. Let $\bblam_1$ be the dual  Recall the dual function \eqref{eq_dual} is a quadratic function 
\begin{equation}
    g(\bblam) = -0.5 \bblam^\top \bbQ \bblam
+ \frac{1}{N}\bblam^\top (\textbf{1} - \bbepsilon) + m(\ccalX, \ccalW), 
\end{equation}
with 
\begin{equation}
    \bbQ_{nm} = \frac{y_n y_m}{N^2}\int_{\ccalX \times \ccalW} k(\bbx_n, \bbs; w) k(\bbx_m, \bbs; w) d\bbs \, dw
\end{equation}
The matrix $\bbQ$ can be divided into sub-matrices based on the agents, to which the kernels centers belong:
\begin{equation}
\bbQ = \begin{bmatrix} 
\bbQ_{11} & \bbQ_{12} \\
\bbQ_{21} & \bbQ_{22} 
\end{bmatrix},
\end{equation}
for which 
\begin{equation}
\begin{aligned}
 \bbQ_{ij(nm)} &= \frac{1}{N^2}\int_{\ccalX \times \ccalW} k(\bbx_n^{(i)}, \bbs; w) k(\bbx_m^{(j)},  \bbs; w) d\bbs \, dw, \\
 &\bbx_n^{(i)} \in \bbX_i.
\end{aligned}
\end{equation}
Then notice that 
\begin{equation}
\begin{aligned}
    \bblam^\top \bbQ \bblam &= \left(\frac{N}{N_1}\bblam_1^\top\right) \bbQ_{11} \left(\frac{N}{N_1}\bblam_1 \right)  \\
    &+ \left(\frac{N}{N_2}\bblam_2^\top\right) \bbQ_{22} \left(\frac{N}{N_2}\bblam_2 \right) \\
    &+ 2 \left(\frac{N}{N_1}\bblam_1^\top\right) \bbQ_{12} \left(\frac{N}{N_2}\bblam_2 \right) \\
    &= \bblam_1^\top \bbQ_1 \bblam_1 + \bblam_2^\top \bbQ_2 \bblam_2 \\
    &+ 2 \left(\frac{N}{N_1}\bblam_1^\top\right) \bbQ_{12} \left(\frac{N}{N_2}\bblam_2 \right)
    \end{aligned}
\end{equation}
Additionally, the measure of the support of the dual function is equal to the sum of the measures of the individual agents
\begin{equation}\label{eq_support}
 m(\ccalX, \ccalW) = m(\ccalX_1, \ccalW) + m(\ccalX_2, \ccalW).
\end{equation}
Then we can conclude the following
\begin{equation}
\begin{aligned}
&\left| g(\bblam) - (g_1(\bblam_1) + g_2(\bblam_2)) \right| \\
&=  \left|2 \left(\frac{N}{N_1}\bblam_1^\top\right) \bbQ_{12} \left(\frac{N}{N_2}\bblam_2 \right) \right|\\
& = \frac{2}{N_1 N_2}\bblam_1^\top \int k(\bbX_1,\bbs;w) k(\bbX_2, \bbs;w) d\bbs dw \bblam_2 \\
& \leq \frac{2 \xi m(\calX,\calW)}{N_1 N_2} \bblam_1^\top \bbJ \bblam_2
\end{aligned}
\end{equation}
The last inequality stems from \eqref{eq_assum} and the fact that a value of a kernel is at most 1. This result can be extended to multiple agents by considering all pairs of $\bbQ_{ij}$ in the difference between the global dual function and the local dual functions.
Therefore we obtain:
\begin{equation}
\left| g(\bblam) - \sum_i g_i(\bblam_i) \right| \leq \frac{2 \xi m L}{N^2}
\end{equation}
for $L = (N^2/(N_i N_j))\bblam_i^\top \bbJ \bblam_j$.
\end{proof}

\section{Proof of Theorem \ref{theorem:separate}}
\label{A:theorem_sep}
\begin{proof}
Recall the relationship between the dual functions of the two problems \eqref{eqn_dual_distinct}. The relationship between the dual functions optimal values can be obtained through triangle inequality
\begin{equation}\label{eqn_distinct_optimal}
\vert g(\bblam^*) - \sum_i g_i(\bblam_i^*) \vert \leq \frac{4 \xi m L}{N^2}, 
\end{equation}
The dual function is strongly concave near the optimal value such that we can establish the relationship between the dual optimal variables
\begin{equation}
    \left\Vert \bblam^* - \bblam_a\right\Vert^2 \leq \frac{2}{\mu}\left\vert g(\bblam^*) - g(\bblam_a^*)\right\vert - \frac{2}{\mu} \nabla g(\bblam^*)(\bblam^* - \bblam_a^*),
\end{equation}
where $\bblam_a^* = [(N_1/N)\bblam_1^\top, \dots, (N_K/N)\bblam_K^*]$.
The gradient at the optimal value $\nabla g(\bblam^*) = 0$ or $\bblam^*= 0$. Therefore, the equation can be reduced to
\begin{equation}
    \left\Vert \bblam^* - \bblam_a\right\Vert^2 \leq \frac{2}{\mu}\left\vert g(\bblam^*) - g(\bblam_a^*)\right\vert \leq \frac{8 \xi m L}{\mu N^2}, 
\end{equation}
The optimal primal value can be obtained from the dual value and therefore we can establish the following inequality
\begin{equation}
\begin{aligned}
    &\left\vert \alpha^*(\bbs,w) - \sum_i \alpha_i^*(\bbs,w) \right\vert =  \\
    &\left\vert \frac{1}{N}\sum_{n=1}^N \bblam^*_n y_n k(\bbx_n,s;w) 
    - \sum_i \frac{1}{N_i} \sum_{\bbx_n \in \bbX_i} \bblam_{a,n}^* y_n k(\bbx,s;w) \right\vert \\
    &\leq \sum_{n=1}^N\left\vert y_n k(\bbx,\bbs;w) \left(\frac{1}{N} \bblam_n^* - \frac{1}{N_i} \bblam_{a,n}^*\right) \right\vert \\
    &\leq \sum_{n=1}^N\left\vert \frac{1}{N} \bblam_n^* - \frac{1}{N_i} \bblam_{a,n}^* \right\vert \\
    &\leq \frac{1}{\sqrt{N}} \Vert \bblam^* - \bblam_a^*\Vert \\
    & \leq \frac{2\sqrt{2 \xi m L}}{N \sqrt{\mu N}}
\end{aligned}
\end{equation}
where $\alpha^*(\bbs,w)$ represents the optimal variable learned by the centralized learner \eqref{eqn_the_central_problem} and $\alpha_i^*(\bbs,w)$ represents the optimal variable learned by the agent \eqref{eqn_the_agent_problem}.
\end{proof}

\section{proof of theorem \ref{th_overlap}}
\label{A:th_overlap}
\begin{proof}
In order to prove the theorem we first formulate the Lagrangian.
\begin{equation}\label{ERM}
\begin{aligned}
\calL_i(\alpha,\bblam) = \rho(\alpha) + \frac{1}{N_i} \sum_{x\in \calX_i} \lambda(x)\left[ \ell(f(x), y) - \epsilon(x) \right]
\end{aligned}
\end{equation}
Similarly, we can construct a function $\calL(f,\lambda)$ which is not associated with any primal function
\begin{equation}\label{SLM}
\calL(\alpha,\bblam) = \rho(\alpha) +  \int_{x\in \calX} \lambda(x)\left[ \ell(f(x), y) - \epsilon(x) \right] p(x) dx
\end{equation}
Notice that the integral is precisely the expected value:
  \[\mathbb{E}_x\left(\lambda(x)\left[ \ell(f(x), y) - \epsilon(x) \right] \right)\]
Therefore the minimization of \eqref{ERM} can be viewed as an empirical risk minimization problem which approximates the statistical loss minimization problem in \eqref{SLM}.
From \cite{lugosi2004bayes} and \cite{cortes2009new} it follows that:
\begin{equation}\label{SLM-ERM}
 \vert \calL(\alpha,\lambda) - \calL_i(\alpha,\bblam) \vert \leq \frac{M}{\sqrt{N_i}}
\end{equation}
where $ M$ is a constant, such that  $\Vert f \Vert^2 \leq M$ and $N_i$ is the sample size. We can construct the dual function and a function based on $\calL(\alpha, \lambda)$
\begin{equation}\label{dual}
g_i(\bblam) = \min_\alpha \calL(\alpha,\bblam)
\end{equation}
\begin{equation}\label{ERM_dual}
g(\lambda) = \min_\alpha \calL(\alpha, \lambda)
\end{equation}
For which we can compute the optimal $\lambda$
\begin{equation}\label{optimal_dual}
\bblam_i^* = \argmax_{\bblam \geq 0} g_i(\bblam)
\end{equation}
\begin{equation}
\lambda^* = \argmax_{\lambda \geq 0} g(\lambda).
\end{equation}
Since the difference between $\calL_i(\alpha,\bblam)$ and $\calL(\alpha,\lambda)$ is bounded, so is there minimums,
\begin{equation}
\vert g(\lambda) - g_i(\bblam) \vert \leq \frac{M}{\sqrt{N_i}}
\end{equation}
Because the inequality holds for any $\lambda$, it must hold for the optimal values.
Let $\lambda_s^*$ be the optimal function $\lambda(\bbx)$ evaluated at the sample points and $\bblam_i^*$ be the optimal dual variable of \eqref{eqn_the_agent_problem}, then by the triangle inequality it follows that:
\begin{equation}
\vert g(\lambda_s^*) - g_i(\bblam_i^*) \vert \leq \frac{2M}{\sqrt{N_i}}
\end{equation}
Furthermore, the dual function is strongly convex near the optimal value:
\begin{equation}
g_i(\bblam_i^*) - g_i(\bblam_s^*) \geq \frac{\mu}{2} \Vert\bblam_i^*- \bblam_s^*\Vert^2 + \nabla g_i(\bblam_i^*) (\bblam_i^* - \bblam_s^*),
\end{equation} 
where $\bblam_s^*$ is a vector for which $\bblam_{s,n}^* = \lambda_s^*(\bbx_n)$.
Notice that the term $\nabla g_i(\bblam_i^*) (\bblam_i^* - \lambda_s^*) \geq 0$. Most of the terms of the gradient are zero since $\bblam_ i$ maximizes the dual function. For the other terms, $\nabla g(\bblam_i^*)_n < 0$ only if $\bblam_{i,n}^* = 0$. In the latter case $(\bblam_{i,n}^* - \lambda_{s,n}^*) \leq 0$. Then we can conclude
\begin{equation}
\Vert\bblam_i^*- \lambda_s^*\Vert^2 \leq \frac{4M}{\mu\sqrt{N_i}}. 
\end{equation}
Next a bound on the functions $\alpha$ can be established which are defined as
\begin{equation}
\alpha_i^*(s,w)  =
\begin{cases} \frac{1}{N_i} \sum_j \lambda_j^* y_j k(\bbx_j,s;w), & \vert \alpha_i(s,w) \vert > \sqrt{2\gamma}\\
0  & otherwise. \\
\end{cases}
\end{equation}
Similarly, a function $\alpha(s,w) = \underset{\alpha \in L_2}{\text{argmin}} \calL(\alpha, \lambda^*)$ can be  be computes as:
\begin{equation}
\alpha(s,w) = \begin{cases}
\int \lambda(x) y(x) k(x,s;w) p(x)\, dx, &  \alpha^2(s,w)  > 2\gamma\\
0 & otherwise.
\end{cases}
\end{equation}
The difference between $\alpha$ and $\alpha_i$ is bounded as follows
\begin{equation}\label{eq_alpha_dif}
\begin{aligned}
&\vert \alpha(s,w) - \alpha_i(s,w) \vert = \\ 
& \left| \int \lambda(x) y_x k(x,s;w) p(x)\, dx - \frac{1}{N_i} \sum_j \lambda_j^* y_j k(\bbx_j,s;w) \right| \leq \\
& \frac{1}{N_i} \left| \sum_j (\bblam_{i,j}^* - \lambda_{s,j})y_j k(\bbx_j,s;w)   \right| + \\
&\left| \int \lambda(x) y_x k(x,s;w) p(x)\, dx - \frac{1}{N_i} \sum_j \lambda_{s,j}^* y_j k(\bbx_j,s;w) \right| \leq\\
& \frac{1}{N_i} \sum_j \vert \bblam_{i,j}^* - \lambda_{s,j} \vert + \frac{c \rho}{\sigma^3 \sqrt{N_i}}\\
& \leq \frac{1}{\sqrt{N_i}} \Vert \bblam_i^* - \lambda_s^* \Vert_2 + \frac{c \rho}{\sigma^3 \sqrt{N_i}} \leq \frac{2\sqrt{M}}{\sqrt{N_i^{1.5} \mu}} + \frac{c \rho}{\sigma^3 \sqrt{N_i}},
\end{aligned} 
\end{equation}
for which $c>0$ is a constant, $\rho = \mathbb{E}_\bbx\left[ \left| \lambda(\bbx) y_\bbx k(\bbx,s;w) \right|^3\right]$ and $\sigma^2 = \mathbb{E}_\bbx\left[ \left| \lambda(\bbx) y_\bbx k(\bbx,s;w) \right|^2\right]$
Given two models trained on independently drawn data sets, with optimal variables $\alpha_1$ and $\alpha_2$ respectively, then the absolute difference between the two variables is

\begin{equation}
\begin{aligned}
&\vert \alpha_1(s,w) - \alpha_2(s,w) \vert \\
&\leq \frac{2 \sqrt{M}}{\sqrt{N_1^{1.5} \mu_1}} + \frac{c \rho}{\sigma^3 \sqrt{N_1}} + \frac{2\sqrt{M}}{\sqrt{N_2^{1.5} \mu_2}} + \frac{c \rho}{\sigma^3 \sqrt{N_2}} \\
&\leq 2 \left(2\sqrt{\frac{M}{\mu N^{1.5}}} + \frac{c \rho}{\sigma^3 \sqrt{N}} \right),
\end{aligned}
\end{equation}
where $N = min(N_i, N_j)$.
\end{proof}

\bibliographystyle{IEEEbib}
\bibliography{kernelBounds}

\end{document}